\documentclass[10pt,twocolumn]{article} 

\usepackage{arxiv}

\RequirePackage[colorlinks,citecolor=blue,linkcolor=blue,urlcolor=blue,pagebackref]{hyperref}

\usepackage{bm}
\usepackage{bbm}
\usepackage{comment}         
\usepackage{multicol}
\usepackage{multirow}
\usepackage{caption}
\usepackage{natbib}

\usepackage{amsmath,amssymb,amsthm}

\theoremstyle{plain}

\newtheorem{theorem}{Theorem}[section]
\newtheorem{lemma}[theorem]{Lemma}

\newtheorem{corollary}[theorem]{Corollary}

\newtheorem{definition}[theorem]{Definition}
\newtheorem{assumption}[theorem]{Assumption}

\newtheorem*{remark}{Remark}

\usepackage{color}

\usepackage{graphicx}
\newcommand{\indep}{\rotatebox[origin=c]{90}{$\models$}}

\newcommand{\R}{\mathbb{R}}

\DeclareMathOperator*{\argmin}{arg\,min}

\newcommand{\pa}{\mathrm{\pa}}

\newcommand{\RN}[1]{%
  \textup{\uppercase\expandafter{\romannumeral#1}}%
}

\usepackage{algorithm}
\usepackage{algorithmic}

\usepackage{authblk}

\usepackage{enumitem}

\title{CATE Lasso: Conditional Average Treatment Effect Estimation\\ with High-Dimensional Linear Regression}

\begin{document}

\author[1]{Masahiro Kato}
\author[1,2]{Masaaki Imaizumi}

\affil[1]{Department of Basic Science, The University of Tokyo}
\affil[2]{RIKEN Center for Advanced Intelligence Project}

\maketitle
\thispagestyle{empty}

\begin{abstract}
In causal inference about two treatments, \emph{Conditional Average Treatment Effects} (CATEs) play an important role as a quantity representing an individualized causal effect, defined as a difference between the expected outcomes of the two treatments conditioned on covariates. This study assumes two linear regression models between a potential outcome and covariates of the two treatments and defines CATEs as a difference between the linear regression models. Then, we propose a method for consistently estimating CATEs even under high-dimensional and non-sparse parameters. In our study, we demonstrate that desirable theoretical properties, such as consistency, remain attainable even without assuming sparsity explicitly if we assume a weaker assumption called \emph{implicit sparsity} originating from the definition of CATEs. In this assumption, we suppose that parameters of linear models in potential outcomes can be divided into treatment-specific and common parameters, where the treatment-specific parameters take difference values between each linear regression model, while the common parameters remain identical. Thus, in a difference between two linear regression models, the common parameters disappear, leaving only differences in the treatment-specific parameters. Consequently, the non-zero parameters in CATEs correspond to the differences in the treatment-specific parameters. Leveraging this assumption, we develop a Lasso regression method specialized for CATE estimation and present that the estimator is consistent. Finally, we confirm the soundness of the proposed method by simulation studies.
\end{abstract}

\section{Introduction} 
Estimating the causal effects of binary treatment from observations is a central task in various fields, such as economics \citep{Wager2018}, medicine \citep{Assmann2000}, and online advertisement \citep{Bottou13}. Specifically, we consider a case where there is a binary treatment and investigate conditional average treatment effects \citep[CATEs,][]{hahn1998role,Heckman1997,Abrevaya2015}, defined as a difference of the expected values of binary treatments' scalar outcomes conditioned on covariates. CATEs have garnered attention as a quantity that captures the heterogeneity among individuals' treatment effects in the population. 

Estimating CATEs requires regression models that approximate the conditional expectation of outcomes.
This study assumes two linear regression models between a potential outcome and covariates of each treatment. Then, CATEs are defined as a difference of the two linear regression models. Under this setting, our interest lies in consistently estimating CATEs when the two linear regression models have high-dimensional and non-sparse parameters.

Estimating parameters in high-dimensional regression models is usually challenging, especially when the dimension is larger than the sample size. The solutions of least squares may lack desirable properties such as consistency, typically held in low-dimensional models under mild conditions. Various estimation and inference approaches have been proposed for high-dimensional linear regression models. This study focuses on linear regression with sparsity using the Lasso \citep{tibshirani96regression, zhao06a, vandeGeer2008}. 

We first formulate our problem using the Neyman-Rubin causal models \citep{Neyman1923, Rubin1974}, which define a potential outcome for each treatment. We observe one of the potential outcomes corresponding to our actual treatment. For each treatment, between the potential outcome and covariates, we assume linear regression models; that is, there are two linear models corresponding to binary treatments. Then, we define CATEs as a difference between the two linear models.

For the linear models of each treatment, we assume that parameters are separable into treatment-specific and common parameters. While the treatment-specific parameters take different values between linear models in each treatment, the common parameters are the same. Therefore, when taking the difference between two linear regression models for each binary treatment, the common parameters disappear, and only the treatment-specific parameters remain. Therefore, even under high-dimensional and non-sparse linear regression models, the total dimension of CATE linear models depends only on that of the treatment-specific parameters. Hence, if the dimension of the treatment-specific parameters is low, we can employ properties similar to ones obtained under sparsity. Because we do not explicitly assume sparsity for linear regression models in potential outcome level, we refer to this sparsity as \emph{implicit sparsity}.

By utilizing this implicit sparsity, we propose the Lasso-based regression specialized for CATE estimation, referred to as the \emph{CATE Lasso}. The CATE Lasso regularizes the parameters by adding an $\ell_1$-norm for differences of the parameters, while the standard Lasso regularizes the parameters themselves. Surprisingly, even if linear models are high-dimension and non-sparse in a potential outcome level, we can still show consistency by utilizing the assumption. Furthermore, this assumption includes a case where linear models in potential outcomes are sparse. Thus, we develop a regression method for CATE estimation with high-dimensional and non-sparse liner models in a potential outcome level under the implicit sparsity assumption.

In summary, our contribution lies in the proposal of the Lasso method specialized for CATE estimation. Our method allows us to estimate the CATE without assuming the sparsity for linear model in a potential outcome level. We also show the consistency of our proposed estimator. Surprisingly, although we cannot show consistency for each linear model of each treatment, we can show consistency for the CATE defined as a difference of the linear models. Our method does not require nuisance estimators such as the propensity score and conditional mean function as in the inverse probability weighting (IPW) and doubly robust (DR) estimators. For that points, our method has advantages compared to the IPW and DR-based methods.

\textbf{Organization.} 
The structure of this study is organized as follows. Initially, we formulate our problem in Section~\ref{sec:problem} and summarize the notations in Section~\ref{sec:notation}. Section~\ref{sec:high} is dedicated to defining our high-dimensional linear regression models with implicit sparsity. For these linear regression models, we introduce Lasso-based estimators in Section~\ref{sec:cateLasso}, referred to as the CATE Lasso estimator. Subsequently, Section~\ref{sec:theoretical} provides theoretical results for the CATE Lasso estimator. Experimental results that confirm the validity of our proposed estimators are showcased in Section~\ref{sec:exp}. In Section~\ref{sec:related}, we introduce related work.

\section{Problem Setting}
\label{sec:problem}
In this section, we define our problem setting. Our formulation is based on the Neyman-Rubin potential outcome framework \citep{Neyman1923,Rubin1974}, which defines potential outcomes and observations, separately. 

\subsection{Potential Outcomes}
Suppose that there are two treatments $\mathcal{D} := \{1, 0\}$. In a typical situation, treatment $d=1$ corresponds to the active treatment, and treatment $d=0$ corresponds to the control treatment. For example, in a clinical trial, $d=1$ corresponds to a new drug, and $d=0$ corresponds to the placebo. We then posit the existence of potential outcome random variables $Y^1, Y^0 \in\mathbb{R}$ corresponding respectively to the treatment $d=1$ and $d=0$. Additionally, we assume that there are $p$-dimensional covariates $X\in\mathcal{X}\subset\mathbb{R}^p$, where $\mathcal{X}$ is the covariate space.

\subsection{CATE}
Let $P := (P^1, P^0)$ denote a set of distributions of $P^1$ and $P^0$, which are joint distributions of $(Y^1, X)$ and $(Y^0, X)$. Let $\mathbb{E}_P[\cdot]$ be an expectation under $P$. 

In this study, our interest lies in the CATE at $X = x \in \mathcal{X}$ defined as
\begin{align*}
    f_P(x) = \mathbb{E}_{P}[Y^1 - Y^0 | X = x] 
\end{align*}
This quantity has been widely used in empirical studies of various fields, such as epidemiology, economics, and political science, because it captures heterogeneous treatment effects for each individual represented by a characteristic $x\in\mathcal{X}$. 

\subsection{Observations}
Although there exist $Y^1$ and $Y^0$ as potential outcomes, we can only observe one of the outcomes, corresponding to an actual treatment $D\in\mathcal{D}$. Let $D \in \mathcal{D}$ denote an actual treatment indicator. By using $Y^1$, $Y^0$, and $D$, we define an observed outcome $Y\in\mathbb{R}$ as 
\begin{align*}
Y = D Y^1 + (1 - D) Y^0,
\end{align*}
Here, if $D = 1$, we observe $Y = Y^1$; if $D = 0$, we observe $Y = Y^0$.

Then, we define our observations. Let $n$ be the sample size. For each $i \in [n] := \{1,2,\dots,n\}$, let $(Y_i, D_i, X_i)$ be an independent and identically distributed (i.i.d.) copy of $(Y, D, X)$.
Then, we suppose that the following $n$ samples are observable:
\begin{align*}
\left\{(Y_i, D_i, X_i)\right\}^n_{i=1}.
\end{align*}

Recall that we defined $P$ as a set of $P^1$ and $P^0$, which are joint distributions of potential outcomes and covariates, $(Y^1, X)$ and $(Y^0, X)$. Similarly, we denote a joint distribution of observations $(Y, D, X)$ by $Q$. For the data-generating process, let $P_0$ and $Q_0$ be sets of the true distributions that generates the observations $\left\{(Y_i, D_i, X_i)\right\}^n_{i=1}$. Let $f_{P_0}$ be $f_0$. 

Our goal is to estimate the CATE $f_0(x)$ by using $\left\{(Y_i, D_i, X_i)\right\}^n_{i=1}$. In particular, we aim to obtain a consistent estimator of the CATE $f_0(x)$ at a point $X = x$.

For identification of $\tau(x)$, we assume the \emph{unconfoundedness}.
\begin{assumption}[Unconfoundedness] \label{asmp:unconfounded}
The treatment indicator $D$ is independent of the potential outcomes for $\{Y^1, Y^0\}$ conditional on $X$:
\[\{Y^1, Y^0\}\indep ~D \ |\ X\]
\end{assumption}
Assumption \ref{asmp:unconfounded} expresses a setting wherein the assignment is independent of the output conditioned on the covariates.
This is the standard approach in treatment effect estimation \citep{Rosenbaum1983}.

\begin{assumption}[Overlap of assignment support] \label{asmp:coherent}
For some universal constant $0 < \varphi < 0.5$, we have
\begin{align*}
    \varphi < \mathbb{E}[D = 1| X ] < 1 - \varphi,\qquad \mathrm{a.s.}
\end{align*}
\end{assumption}
Assumption~\ref{asmp:coherent} allows us to avoid overlap in treatment assignments.
By this assumption, we can guarantee the identifiability of the assignment and the parameters \citep{imbens_rubin_2015}.

\section{Notation} 
\label{sec:notation}
Let $\mathcal{P} := \{1,\dots,p\}$. 
Let us define a ($n\times p$)-matrix 
$
   \bm{X} := \begin{pmatrix}
    X^\top_{1} \\
    X^\top_{2} \\
    \vdots \\
    X^\top_{n} 
\end{pmatrix} = \begin{pmatrix}
    X_{1, 1} & X_{1, 2} & \cdots & X_{1, p}\\
    X_{2, 1} & X_{2, 2} & \cdots & X_{2, p}\\
    \vdots & \vdots & \ddots & \vdots\\
    X_{n, 1} & X_{n, 2} & \cdots & X_{n, p}
\end{pmatrix}$,  and a $n$-dimensional vector $
\mathbb{D} := (D_1,D_2,...,D_n)^\top
$. 
Let $m^1 := \sum^n_{i=1}D_i$ and $m^0 := \sum^n_{i=1}( 1- D_i) = n - m_1$. 
For each $d\in\mathcal{D}$ and $k\in\mathcal{P}$, let us define $m^d$-dimensional column vectors $\widetilde{\mathbb{Y}}^d$ and $\widetilde{\mathbb{X}}^d_k$ as 
$
    \widetilde{\mathbb{Y}}^d = (Y_i)^\top_{i\in\{1,2,\dots, n\}| D_i=d}$, $
    \widetilde{\mathbb{X}}^d_k = (X_{i,k})^\top_{i\in\{1,2,\dots, n\}| D_i=d}
$. 
Let us also define a ($m^d \times p$)-matrix $\widetilde{\bm{X}}^d$ as $    \widetilde{\bm{X}}^d = \begin{pmatrix}
        \widetilde{\mathbb{X}}^d_1 & \widetilde{\mathbb{X}}^d_2 & \cdots & \widetilde{\mathbb{X}}^d_p
    \end{pmatrix}$. 
We also denote the ($i,j$)-element of $\widetilde{\bm{X}}^d$ as $\widetilde{X}^d_{i, j}$. For a vector $v = (v_1,\dots, v_p)\in\mathbb{R}^p$, we denote its $\ell_q$ norm by $\|v\|_q := (\sum^p_{i=1}|v_i|^q)$ for $q \geq 1$.

\section{High-Dimensional Linear Regression with Implicit Sparsity}
\label{sec:high}
This section introduces high-dimensional linear regression models with implicit sparsity.

\subsection{Potential High-Dimensional Linear Regression Models}
In this study, we consider a linear relationship between a potential outcome $Y^d$ ($d\in\mathcal{D}$) and covariates $X$. For each $d \in \mathcal{D}$ and $P$, we posit the following linear regression model:
\begin{align}
\label{eq:linear}
    Y^d = X^\top\bm{\beta}^d + \varepsilon^d, 
\end{align}
where $\bm{\beta}^d \in \mathbb{R}^p$ is a $p$-dimensional parameter, while $\varepsilon^d$ is an independent noise variable with a zero mean and finite variance. In our analysis, we make the following assumptions about $\varepsilon^d$. 
\begin{assumption}
\label{asm:eror}
For the error term $\epsilon^d$, $\mathbb{E}[\varepsilon^d | X] = 0$ holds a.s. Furthermore, we assume that $\varepsilon^1$ and $\varepsilon^0$ are independent each other. Additionally, the variance is 
finite; that is, $(\sigma^d_\varepsilon)^2 := \mathbb{E}[(\varepsilon^d)^2] < C < \infty$ for some universal constant $C$. 
\end{assumption}

Let $\varepsilon^d_i$ be an i.i.d. copy of $\varepsilon^d$, and $\bm{\varepsilon}^d$ be $(\varepsilon^d_1\ \cdots\ \varepsilon^d_n)^\top$.   

We also allow for our linear regression models to be high-dimensional; that is, $p > n$. In such a situation, a common approach to obtaining a consistent estimator is to assume that $\bm{\beta}^d$ has the \textit{sparsity}, that is, most of the elements of $\bm{\beta}^d$ is zero. 
However, this study does not assume sparsity directly to potential outcomes; instead, we leverage the property of CATE that emerges from the differentiation between $Y^1$ and $Y^0$. 

\subsection{Individual and Common Parameters}
\label{sec:ind_common}
Our key assumption for the CATE estimation is that the data generating models for $Y^1$ and $Y^0$ have several parameters in common.
These common parameters plays an important role for a sparsity-like property, without the explicit sparsity assumption. 

Specifically, we suppose that the parameter $\bm{\beta}^d_0$ can be separated into two parameters $\bm{\alpha}^d_0$ and $\bm{\gamma}_0$. 
\begin{assumption}[Separability]
\label{asm:separable}
For parameters in linear models in \eqref{eq:linear}, under the true distribution $P_0$, 
\begin{align}
    \bm{\beta}^d_0 = \left(\big(\bm{\alpha}^d_0\big)^{\top}, \bm{\gamma}^\top_0\right)^\top, \label{def:separability}
\end{align}
holds, where $\bm{\alpha}^d_0 \in\mathbb{R}^{s_0}$ is an $s_0$-dimensional vector defined for each $d \in \{1, 0\}$, and $\bm{\gamma}_0$ is a ($p - s_0$)-dimensional vector. 
\end{assumption}
We refer to $\bm{\alpha}^d_0 \in\mathbb{R}^{s_0}$ as an \emph{individual parameter} and $\bm{\gamma}_0$ as a \emph{common parameter}. 
We also give a covariate form $X=(W^\top, Z^\top)^\top$ with subvectors $W \in \R^{s_0}$ and $Z \in \R^{p - s_0}$, that correspond to $\bm{\alpha}^d_0$ and $\bm{\gamma}_0$, respectively.
Then, the linear regression model \eqref{eq:linear} is rewritten as 
\begin{align*}
    Y^d = W^\top\bm{\alpha}^d_0 + Z^\top \bm{\gamma}_0 + \varepsilon^d
\end{align*}
under $P_0$. 
We assume that we do not know which elements in covariates $X$ correspond to $W$ and $Z$. 

\begin{remark}
The separability assumption is a generalization of the traditional assumption that there are \emph{known} common terms in a regression model. Consider the following models under the separability assumption and $P_0$:
\begin{align*}
    &Y^d = d W^\top\bm{\alpha}^1 + (1 - d) W^\top\bm{\alpha}^0_0 + Z^\top \bm{\gamma}_0 + d\varepsilon^1 + (1-d)\varepsilon^0\\
    &= d W^\top\left(\bm{\alpha}^1_0 - \bm{\alpha}^0_0\right) + W^\top\bm{\alpha}^0_0 + Z^\top \bm{\gamma}_0 + d\varepsilon^1 + (1-d)\varepsilon^0\\
    &= d W^\top\overline{\bm{\alpha}}_0 + R^\top\bm{\rho}_0 + \overline{\varepsilon}^d,
\end{align*}
where $\overline{\bm{\alpha}}_0 = \bm{\alpha}^1_0 - \bm{\alpha}^0_0$, $R = 
   (W^\top\ \ Z^\top) 
^\top$, $\bm{\rho}_0 = 
   ((\bm{\alpha}^0)^\top_0\ \ \bm{\gamma}^\top_0)
^\top$, and $\overline{\varepsilon} = d\varepsilon^1 + (1-d)\varepsilon^0$. If we know which terms are common parameters in advance, we can conduct linear regression based on $Y^d = d W^\top\overline{\bm{\alpha}}_0 + R^\top\bm{\rho}_0 + \overline{\varepsilon}^d$. Existing studies mainly consider such regression models. For example, the R-learner by \citet{Nie2020} considers those models, assuming a more general form than linear regression. In contrast, in our study, we do not assume the knowledge about which terms are common parameters. 
\end{remark}

\subsection{CATE Linear Regression Model with Implicit Sparsity}
\label{sec:implicit}
We firstly consider the following unified linear model derived from the original linear model \eqref{eq:linear} for $d \in \mathcal{D}$:
\begin{align}
\label{eq:regression_model}
\overline{Y} &= X^\top\bm{\beta} + \overline{\varepsilon},
\end{align}
where 
\begin{align}
\overline{Y} &= Y^1 - Y^0,\ \bm{\beta} = \bm{\beta}^1 - \bm{\beta}^0,\ \overline{\varepsilon} = \varepsilon^1 - \varepsilon^0. \label{def:regression_variables}
\end{align}
Here, $\overline{Y}$ is an unobservable random variable defined using potential outcomes, and $\overline{\varepsilon}$ is an error term that is a composite of $\varepsilon^1$, $\varepsilon^0$, and $D$. It should also be noted that $\mathbb{E}[\overline{\varepsilon}|X] = 0$ almost surely.

This model \eqref{eq:regression_model} demonstrates the implicit sparsity. 
By considering the difference between $Y^1$ and $Y^0$ and using \eqref{def:separability},
we can eliminate the term $Z^\top\bm{\gamma}$ from the regression model. Consequently, we obtain the following:
\begin{lemma}[Implicit Sparsity]
    Consider the linear regression model \eqref{eq:linear} for $d \in \mathcal{D}$ and suppose that Assumption \ref{asm:separable} holds.
    Then, under $P_0$, the (unobserved) linear regression model \eqref{eq:regression_model} is given as
\begin{align}
\label{eq:hidden_linear}
\overline{Y} &= W^\top\left(\bm{\alpha}^1_0 - \bm{\alpha}^0_0\right) + \overline{\varepsilon}.
\end{align}
\end{lemma}
\begin{proof}
    It obviously holds by the definitions of the variables \eqref{def:regression_variables} and the form \eqref{def:separability} of $\bm{\beta}^d_0$.
\end{proof}

Note that the linear regression model in \eqref{eq:hidden_linear} only has a $s_0$-dimensional parameter $\bm{\alpha}^1_0 - \bm{\alpha}^0_0$.
This situation is regarded that the $p$-dimensional parameter $\bm{\beta}_0$ in the model \eqref{eq:regression_model} has only $s_0$ non-zero elements, and the other elements are zero which is ignored.
In other words, we can regard $\bm{\beta}_0$ in the model \eqref{eq:regression_model} has the sparsity in spite that we do not assume the sparsity on $\bm{\beta}^1_0$ and $\bm{\beta}^0_0$. 
We refer to this property as the \textit{implicit sparsity}. 

\section{Lasso for CATE Estimation}
\label{sec:cateLasso}
The Lasso \citep{tibshirani96regression} is an estimation method with a regularization that adds the $\ell_1$-norm penalty on the parameters. 
In this section, we propose the CATE Lasso estimator, which utilizes the implicit sparsity introduced in Section~\ref{sec:implicit}. Unlike the standard Lasso  penalizing the parameter itself, the CATE Lasso penalizes a difference between two parameters.

We focus on the fact that for each $d\in\mathcal{D}$, the following linear regression model holds under $P$:
\begin{align*}
    \widetilde{\mathbb{Y}}^d = \widetilde{\bm{X}}^d\bm{\beta}^d + \widetilde{\bm{\varepsilon}}^d,
\end{align*}
where $\widetilde{\bm{\varepsilon}}^d$ is an $m^d$-dimensional vector whose $i$-th element $\widetilde{\bm{\varepsilon}}^d_i$ is equal to $\varepsilon^d_j$ such that $j = \min\{k\in[N]| \sum^k_{s=1} \mathbbm{1}[D_s = d] = i\}$.

Let $\bm{\beta}_0$, $\bm{\alpha}_0$ and $\bm{\gamma}_0$ be the parameter of distributions that generate observations; that is, they are the true parameters of $\bm{\beta}$, $\bm{\alpha}$ and $\bm{\gamma}$ under $P_0$.  Then, under the linear regression models, for $x = \left(w, z\right)$, where $w\in\mathcal{W}$ and $z\in \mathcal{Z}$, the CATE $f_0(x)$ can be rewritten as
\begin{align*}
    f_0(x) = x^\top \bm{\beta}_0 = w^\top \left(\bm{\alpha}^1_0 - \bm{\alpha}^0_0\right).
\end{align*}

We estimate $\bm{\beta}_0$ by using the least squares with $\ell_1$-penalty. Because we cannot observe $Y^1_i - Y^0_i$, we instead minimize two squared losses defined between $Y^1_i$ and $\widetilde{X}^\top_i \bm{\beta}^1$ and between $Y_i^0$ and $\widetilde{X}^\top_i \bm{\beta}^0$. Unlike the standard Lasso, we penalize $\bm{\beta}^1 - \bm{\beta}^0$, not $\bm{\beta}^1$ and $\bm{\beta}^0$, separately. We refer to our Lasso as the \emph{CATE Lasso}.

\subsection{Estimation Strategy of CATEs}
First, we consider a estimation strategy for CATEs under the implicit sparsity. 
In other words, we consider the regularization for the difference between $\bm{\beta}^1$ and $\bm{\beta}^0$ to utilize the implicit sparsity. Let $\widehat{\bm{\beta}}^1$ and $\widehat{\bm{\beta}}^0$ be estimators of $\bm{\beta}_0^1$ and $\bm{\beta}_0^0$, which are defined as
\begin{align}
\label{eq:Lasso_cate}
&\left(\widehat{\bm{\beta}}^1, \widehat{\bm{\beta}}^0\right) := \left(\widehat{\bm{\beta}}^1(\lambda), \widehat{\bm{\beta}}^{0}(\lambda)\right)\\
&\in \argmin_{\bm{\beta}^1\in\mathbb{R}^p,\ \bm{\beta}^0\in\mathbb{R}^p} \Bigg\{\delta\|\widetilde{\mathbb{Y}}^1 - \widetilde{\bm{X}}^1\bm{\beta}^1 \|^2_2 / n\nonumber\\
&+ (1 - \delta)\|\widetilde{\mathbb{Y}}^0 - \widetilde{\bm{X}}^0\bm{\beta}^0 \|^2_2 / n + 2\lambda\left\|\bm{\beta}^1 - \bm{\beta}^0\right\|_1\Bigg\},\nonumber
\end{align}
where $\lambda > 0$ is a penalty coefficient and $\delta\in(0, 1)$ is a balancing weight. 

In this strategy, we regularize a difference $\bm{\beta}^1 - \bm{\beta}^0$, unlike the standard Lasso. This regularization allows us to employ the sparsity in $\bm{\beta}^1 - \bm{\beta}^0$, while we can assume that $\bm{\beta}^d_0$ is \emph{not} sparse as well as the standard Lasso. Similar estimators have been proposed in the literature of Lasso, such as the fused Lasso \citep{Tibshirani2005} and $\ell_1$-trend filtering \citep{Kim2009}.

\subsection{The CATE Lasso}
\label{sec:cate_interpolate}
In the optimization problem \eqref{eq:Lasso_cate}, there can be multiple solutions unlike the standard Lasso because we only impose a regularization for $\bm{\beta}^1 - \bm{\beta}^0$. 
Therefore, this section develops a uniquely determined estimator of $\bm{\beta}^0$.
We refer to this estimator as the \emph{CATE Lasso estimator}, a Lasso specialized for CATE estimation. 

In preparation, we consider the following optimization problem, which is equivalent to the problem \eqref{eq:Lasso_cate}: 
\begin{align}
\label{eq:Lasso_cate_implement}
&\left(\widehat{\bm{\beta}}, \widehat{\bm{\beta}}^0\right) := \left(\widehat{\bm{\beta}}(\lambda), \widehat{\bm{\beta}}^0(\lambda)\right)\\
&\in \argmin_{\bm{\beta}\in\mathbb{R}^p,\ \bm{\beta}^0\in\mathbb{R}^p} \Big\{\delta\|\widetilde{\mathbb{Y}}^1 - \widetilde{\bm{X}}^1\left(\bm{\beta} + \bm{\beta}^0\right) \|^2_2 / n\nonumber\\
&\ \ \ \ \ \  \ + (1-\delta)\|\widetilde{\mathbb{Y}}^0 - \widetilde{\bm{X}}^0\bm{\beta}^0 \|^2_2 / n + 2\lambda\left\|\bm{\beta}\right\|_1\Big\},\nonumber
\end{align}
where $\bm{\beta}$ is a parameter which represents $\bm{\beta}_0 = \bm{\beta}^1_0 - \bm{\beta}^0_0$. 
This problem is mathematically equivalent to our previous definition of $(\widehat{\bm{\beta}}^1, \widehat{\bm{\beta}}^0)$ in \eqref{eq:Lasso_cate}, since we can achieve $\widehat{\bm{\beta}}^1 = \widehat{\bm{\beta}} + \widehat{\bm{\beta}}^0$ by the estimators $(\widehat{\bm{\beta}}, \widehat{\bm{\beta}}^0)$ in \eqref{eq:Lasso_cate_implement}.
Then, we estimate $\bm{\beta}_0$ with the following procedure.

\paragraph{(I) Interpolating estimator of $\bm{\beta}^0_0$.} First, we estimate $\bm{\beta}^0_0$ as
\begin{align*}
&\widehat{\bm{\beta}}^0:=\argmin_{\bm{\beta}^0\in\mathbb{R}^p} \Big\{(1-\delta)\|\widetilde{\mathbb{Y}}^0 - \widetilde{\bm{X}}^0\bm{\beta}^0 \|^2_2 / n\Big\}.
\end{align*}
The estimator $\widehat{\bm{\beta}}^0$ has the following analytical solution:
\begin{align*}
&\widehat{\bm{\beta}}^0 = \Big\{\left(\widetilde{\bm{X}}^0\right)^\top\widetilde{\bm{X}}^0\Big\}^{\dagger}\left(\widetilde{\bm{X}}^0\right)^\top\widetilde{\mathbb{Y}}^0,
\end{align*}
where $ \{(\widetilde{\bm{X}}^0)^\top\widetilde{\bm{X}}^0\}^{\dagger}$ is a pseudoinverse of $(\widetilde{\bm{X}}^0)^\top\widetilde{\bm{X}}^0$. Such an estimator is referred to as an \emph{interpolating estimator} or a \emph{minimum-norm estimator} \citep{Bartlett2020} because $\widehat{\bm{\beta}}^0$ is a parameter with the smallest $L^2$ norm $\bm{\beta}^0$ among ($\|\bm{\beta}^0\|_2$) that perfectly fit $\mathbb{Y}^0$ with $p \geq n$. Note that we cannot identify $\bm{\beta}^0_0$ itself even if using this estimator; however, as shown in Section~\ref{sec:theoretical}, we can use it for estimating $\bm{\beta}_0$. 

\paragraph{(II) Lasso for estimating $\bm{\beta}_0$.} 
Using the interpolating estimator $\widehat{\bm{\beta}}^0$ and the Lasso, we estimate $\bm{\beta}_0$ as
\begin{align*}
&\widehat{\bm{\beta}} :=\argmin_{\bm{\beta}\in\mathbb{R}^p} \left\{\delta\|\widetilde{\mathbb{Y}}^1 - \widetilde{\bm{X}}^1\left(\bm{\beta} + \widehat{\bm{\beta}}^0\right) \|^2_2 / n + 2\lambda\left\|\bm{\beta}\right\|_1\right\},
\end{align*}
for some $\lambda > 0$. 
This optimization can be solved by the standard algorithm (e.g. the coordinate descent method) used for the Lasso. 
Thus, we obtain the CATE Lasso estimator $\widehat{\bm{\beta}}$ for $\bm{\beta}_0$. 

Note that the choice of $\delta$ does not affect the estimation (we can include $\delta$ into $\lambda$). However, in practice, we may stabilize our method by a suitable choice of $\delta$. 

\paragraph{(III) Estimator for CATE $f_0(x)$.} 
Using the CATE Lasso estimator $\widehat{\bm{\beta}}$, we estimate the CATE $f_0(x)$ for each $x\in\mathcal{X}$ as
\begin{align*}
    \widehat{f}(x) := x^\top \widehat{\bm{\beta}}.
\end{align*}

To the best of our knowledge, such a type of estimator is novel in the literature of CATE estimation.




\section{Theoretical Results}
\label{sec:theoretical}
This section provides several theoretical properties of our CATE Lasso estimator. Our analysis is inspired by \citet{vandeGeer2014}. 

First, we consider a case where $\bm{X}$ is non-random (fixed-design), $\varepsilon^d$ follows a sub-Gaussian distribution, and $p$ is fixed. We show more detailed arguments about the analysis for a fixed design in Appendix~\ref{sec:theoretical_fix}. Based on the results of the fixed design, we show consistency of $\widehat{\bm{\beta}}$ under a random design, where $\bm{X}$ is random, $\varepsilon^d$ is non-Gaussian, and $p = p_n \to \infty$ as $n\to\infty$.

To derive a tight upper bound for the estimation error, we follow \citet{Buhlmann2011} in asserting that the compatibility condition plays a crucial role in identifiability. To define the compatibility condition, for a $p\times 1$ vector $\bm{\beta}$ and a subset $\mathcal{S}_0 \subseteq \mathcal{P}$, we introduce a notation that denotes the sparsity of $\bm{\beta}_0$. Let $\mathcal{S}_0 \subseteq \mathcal{P}$ be a set such that
\begin{align*}
    \mathcal{S}_0 := \left\{j: \beta_{0, j} \neq 0\right\},
\end{align*}
where $\beta_{0, j}$ is the $j$-th element of $\bm{\beta}_0$. 
For an index set $\mathcal{S}_0 \subset \mathcal{P}$, let $\bm{\beta}_{\mathcal{S}_0}$ and $\bm{\beta}_{\mathcal{S}^c_0}$ be vectors whose $j$-th elements are defined as follows:
\begin{align*}
    \beta_{j, \mathcal{S}_0} := \beta_{j}\mathbbm{1}[j\in \mathcal{S}_0],\qquad \beta_{j, \mathcal{S}^c_0} := \beta_{j}\mathbbm{1}[j\notin \mathcal{S}_0],
\end{align*}
respectively, where $\beta_j$ is the $j$-th element of $\bm{\beta}$. Therefore, we obtain $\bm{\beta} = \bm{\beta}_{\mathcal{S}_0} + \bm{\beta}_{\mathcal{S}^c_0}$. 

Note that from the definition of the individual and common parameters, it holds that
\begin{align*}
    \beta_{j, \mathcal{S}_0} = \alpha^1_j - \alpha^0_j\quad \mathrm{and}\quad  \beta_{j, \mathcal{S}^c_0} = \gamma_j - \gamma_j = 0;
\end{align*}
that is, $j (\in \mathcal{S}_0)$-th element of $\bm{\beta}^d$ corresponds to $\bm{\alpha}^d$ and $j (\in \mathcal{S}^c_0)$-th element of $\bm{\beta}^d$ corresponds to $\bm{\gamma}$. Here, note that $s_0 = |\mathcal{S}_0|$, where recall that $s_0$ is defined as the dimension of $\bm{\alpha}$ in Section~\ref{sec:ind_common}. 

Let $\widehat{\Sigma}^1 := (\widetilde{\bm{X}}^1)^\top \widetilde{\bm{X}}^1 / n$.
Here, we define the compatibility condition. 
\begin{definition}[Compatibility condition. From (6.4) in \citet{Buhlmann2011}]
    We say that the compatibility condition holds for the set $S_0$, if there is a positive constant $\phi_0 > 0$ such that for all $\bm{\beta}$ satisfying $\|\bm{\beta}_{\mathcal{S}^c_0}\|_1 \leq 3 \|\bm{\beta}_{\mathcal{S}_0}\|_1$, we have
    \begin{align*}
        \|\bm{\beta}_{\mathcal{S}_0}\|^2_1 \leq s_0 \bm{\beta}^\top \widehat{\Sigma}^1 \bm{\beta} / \left(\phi_0\right)^2. 
    \end{align*}
    We refer to $\left(\phi_0\right)^2$ as the compatibility constant.
\end{definition}

Let us denote the diagonal elements of $\widehat{\Sigma}^1$, by $\left(\widehat{\sigma}^1_j\right)^2 := \widehat{\Sigma}^1_{j, j}$, for $j\in\mathcal{P}$. 
Using the compatibility condition, we make the following assumption. 
\begin{assumption}[From (A1) in \citet{vandeGeer2014}]
\label{asm:a1_vandegeer}
    The compatibility condition holds for $\widehat{\Sigma}^1$ with compatibility constant $(\phi_0)^2 > 0$. Furthermore, $\max_{j, d} \left(\widehat{\sigma}^1_j\right)^2 \leq M^2$ holds for some $0 < M < \infty$. 
\end{assumption}
Furthermore, we make the following assumption.
\begin{assumption}
\label{asm:finte_pseudo}
With probability one, each element of 
    $\{(\widetilde{\bm{X}}^0)^\top \widetilde{\bm{X}}^0\}^\dagger(\widetilde{\bm{X}}^1)^\top \widetilde{\bm{X}}^1$ is finite.
\end{assumption}
The matrix 
$\{(\widetilde{\bm{X}}^0)^\top \widetilde{\bm{X}}^0\}^\dagger(\widetilde{\bm{X}}^1)^\top \widetilde{\bm{X}}^1$ 
represents a discrepancy between $\widetilde{\bm{X}}^0$ and $\widetilde{\bm{X}}^1$. 
When $\widetilde{\bm{X}}^0$ and $\widetilde{\bm{X}}^1$ are identical or having common eigenvectors with similar eigenvalues, the matrix is approximately an identity, hence Assumption \ref{asm:finte_pseudo} is satisfied. 
It represents the correlation between treatments and covariates, thus explains the structure specific to CATE.

Then, we derive an oracle inequality and consistency result for the CATE Lasso estimator. The proof is shown in Appendix~\ref{appdx:oracle_inequality}.
\begin{theorem}[Oracle inequality and consistency]
\label{thm:oracle}
Assume a linear model in \eqref{eq:linear} with fixed design for $\widetilde{\bm{X}}$ and $\mathbb{D}$, which satisfies Assumptions~\ref{asm:eror}--\ref{asm:separable} and \ref{asm:a1_vandegeer}--\ref{asm:finte_pseudo}. Also suppose that $\varepsilon^d$ follows a centered sub-Gaussian distribution  with variance $(\sigma^d_\varepsilon)^2$. Let $t>0$ be arbitrary. Consider the CATE Lasso estimator $\widehat{\bm{\beta}}$ with regularization parameter $\lambda \geq 3 M \sigma^\dagger_{\varepsilon}\sqrt{\frac{2\left(t^2 + \log p\right)}{n}}$, where $(\sigma^\dagger_{\varepsilon})^2 = (\sigma^1_{\varepsilon})^2 + (\sigma^0_{\varepsilon})^2\mathrm{tr}(\{(\widetilde{\bm{X}}^0)^\top \widetilde{\bm{X}}^0\}^\dagger(\widetilde{\bm{X}}^1)^\top \widetilde{\bm{X}}^1)$. Then, with probability at least $1 - 2\exp\left(-t^2\right)$, 
        \begin{align*}
            &\left\| \widehat{\bm{\beta}} - \bm{\beta}_0 \right\|_1 \leq C_1 \lambda s_0 / \left( \phi_0 \right)^2,\mbox{~~and}\\
            &\left\|\widetilde{\bm{X}}^1\left(\widehat{\bm{\beta}} - \bm{\beta}_0\right) \right\|^2_2 / n \leq C_2 \lambda^2 s_0 / \left( \phi_0 \right)^2,
        \end{align*}
        hold, where $C_1, C_2 > 0$ are some universal constants.
\end{theorem}
This theorem implies that under a proper $t$ and penalty $\lambda$ (oracle), we can show the convergence of the estimator $\widehat{\bm{\beta}}$ to the true value $\bm{\beta}_0$ with high probability. 

Let $\Sigma^1$ be the population of $\widehat{\Sigma}^1$. 
Based on the results of a fixed-design, we show the consistency of $\widehat{\bm{\beta}}$ under a random design where the covariates $X_i$ and treatment assignments $D_i$ are random and $p = p_n \to \infty$ as $n\to \infty$. By assuming $\widetilde{X}^d$ is sub-Gaussian, we can obtain the following theorem. 
\begin{theorem}
\label{thm:nongaussian}
Suppose that Assumptions~\ref{asmp:unconfounded}--\ref{asmp:coherent} hold. 
Assume a linear model in \eqref{eq:linear} with Assumptions~\ref{asm:eror}--\ref{asm:separable} and \ref{asm:a1_vandegeer}--\ref{asm:finte_pseudo}. If $\widetilde{\bm{X}}^d$ is sub-Gaussian and ${\Sigma}^1$ has a strictly positive smallest eigenvalue $\Lambda^2_{\min}$, satisfying $\max \Sigma^1_{j,j} = O(1)$  and $1/\Lambda^2_{\min} = O(1)$ as $n\to \infty$. Consider the CATE Lasso estimator $\widehat{\bm{\beta}}$ and $\lambda \asymp \sqrt{\log(p) / n}$. If $s_0 = o(\sqrt{\log(p)/n})$ holds, then we have the following as $n \to \infty$:
        \begin{align*}
            &\left\| \widehat{\bm{\beta}} - \bm{\beta}_0 \right\|_1 = O_P\left(s_0\sqrt{\log(p) / n}\right),\mbox{~~and}\\
            &\left\|\widetilde{\bm{X}}^1\left(\widehat{\bm{\beta}} - \bm{\beta}_0\right) \right\|^2_2 / n = O_P\left(s_0\log(p) / n\right).
        \end{align*}
\end{theorem}
This result implies consistency of $\widehat{\bm{\beta}}$; that is, $\widehat{\bm{\beta}} \xrightarrow{\mathrm{p}}\bm{\beta}_0$. 
\begin{proof}
    From Theorem~1.6 in \citet{zhou2009restricted} (a sub-Gaussian extension of Theorem~1 in \citet{Raskutti2010}), there is a constant $L=O(1)$ as $n\to \infty$ depending on $\Lambda_{\min}$ only such that with probability tending to one the compatibility condition holds with compatibility constant $(\phi_0)^2 > 1/L^2$. Combining this result and Theorem~\ref{thm:oracle} yields the statement.
\end{proof}

Note that although these results indicate consistency of $\widehat{\bm{\beta}}$, it does not ensure the convergence of $\widehat{\bm{\beta}}^1$ and $\widehat{\bm{\beta}}^0$ to $\bm{\beta}^1_0$ and $\bm{\beta}^0_0$, respectively. In fact, since we only regularize $\bm{\beta}^1 - \bm{\beta}^0$, the minimizers $\widehat{\bm{\beta}}^1$ and $\widehat{\bm{\beta}}^0$ of the objective function might not be unique. This finding suggests that even if we can not estimate $\bm{\beta}^1_0$ and $\bm{\beta}^0_0$ consistently, it is still possible to consistently estimate $\bm{\beta}_0$ under the implicit sparsity assumption.

From Theorem~\ref{thm:nongaussian}, we obtain the following corollary.
\begin{corollary}
    Assume the same conditions in Theorem~\ref{thm:nongaussian}. If $s_0 = o(\sqrt{\log(p)/n})$, then
    $
        \big|x^\top\widehat{\bm{\beta}} - f_P(x) \big| = \big|x^\top\big(\widehat{\bm{\beta}} - \bm{\beta}_0 \big) \big| = o_P(1)
    $
    holds for $x\in\mathcal{X}$ as $n \to \infty$. 
\end{corollary}
We used the Cauchy-Schwarz inequality as $ \big|x^\top\widehat{\bm{\beta}} - f_P(x) \big| \leq \|x\|_2\|\widehat{\bm{\beta}} - \bm{\beta}_0\|_2$, and $\|\widehat{\bm{\beta}} - \bm{\beta}_0\|_2 \leq \|\widehat{\bm{\beta}} - \bm{\beta}_0\|_1$. 

\begin{figure*}[h]
  \centering
  \vspace{-3mm}
    \includegraphics[width=140mm]{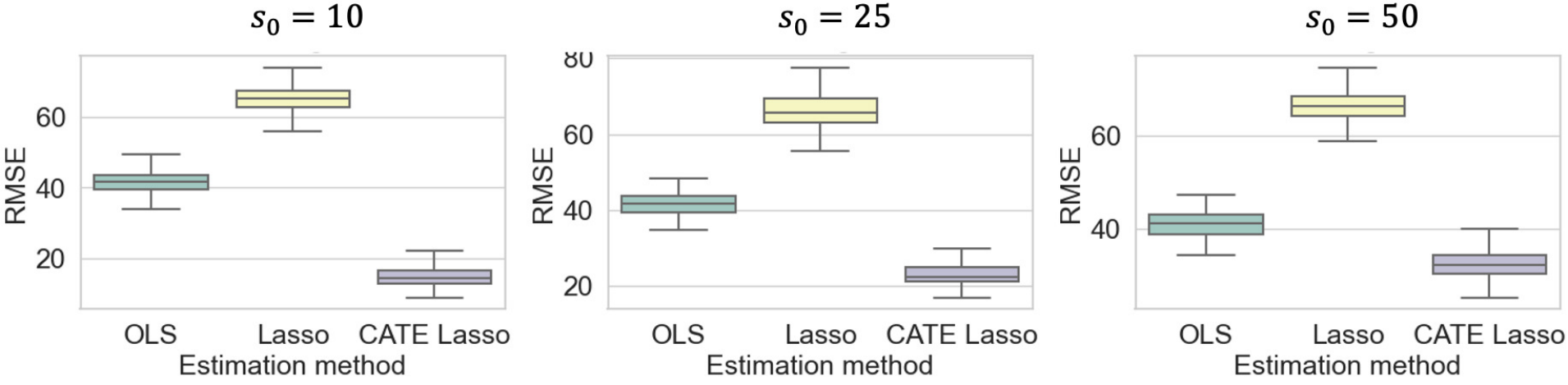}
\vspace{-3mm}
\caption{RMSE of simulation studies with $p = 300$ and $s_0 = 10, 25, 50$.}
\label{fig:exp1}
\centering
    \includegraphics[width=140mm]{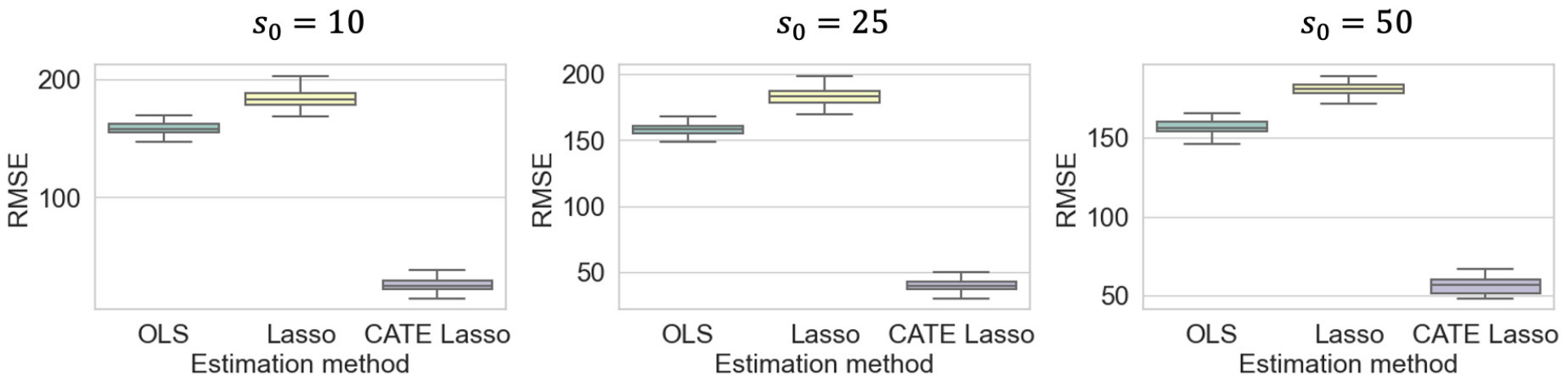}
\vspace{-3mm}
\caption{RMSE of simulation studies with $p = 1000$ and $s_0 = 10, 25, 50$.}
\label{fig:exp2}
\vspace{-5mm}
\end{figure*}

\section{Related Work}
\label{sec:related}
This section briefly introduces related work. More details and open issues, including an extension to the debiased Lasso, are discussed in Appendix~\ref{appdx;related}.

\subsection{Related Work}
Early work on CATEs are \citet{Heckman1997} and \citet{Heckman2005}. 
With the rise of machine learning algorithms, various methods for CATE estimation have been proposed. Some of them are summarized as meta-learners by \citet{Kunzel2019} (see the following section). 
There is a stream of work that employs neural networks \citep{Johansson2016,Shalit2017,Shi2019,Hassanpour2020Learning,curth2021nonparametric,curth2021inductive}, utilizing methods and properties of neural networks, such as representation learning \citep{bengio2014representation} and multi-task learning \citep{Caruana1997}. \citet{yoon2018ganite} applies the generative adversarial nets for CATE estimation. Furthermore, methods that utilize Gaussian processes \citep{Alaa2017,Alaa2018}, deep kernel learning \citep{Zhang2020}, boosting, tree-based methods \citep{Zeileis2008,Su2009,ImaiSt2011,Kang2012,Lipkovich2011,loh2012,Wager2018,Athey2019,Chatla2020}, nearest neighbor matching, series estimation, and Bayesian additive regression trees have been developed \citep{Hill2011}. Numerous methods employing machine learning algorithms have also been proposed \citep{Li2017,Kallus2017,Powers2017,Subbaswamy2018,Zhao2019,Nie2020,Hahn2020}.

\subsection{Meta-learners}
Certain CATE estimators can be categorized into meta-learners. Representatives are listed below:
\begin{description}[topsep=0pt, itemsep=0pt, partopsep=0pt, leftmargin=*]
\item[The S-learner \citep{Kunzel2019}:] This approach estimates $\mathbb{E}[Y | X, D]$. Using the estimator $\widehat{\mathbb{E}}[Y| X, D]$, we estimate the CATE as $\widehat{\mathbb{E}}[Y| X = x, D = 1] - \widehat{\mathbb{E}}[Y| X = x, D = 0]$.
\item[The T-learner \citep{Kunzel2019}:] This method consists of a two-step procedure: in the first stage, we separately estimate the parameters of linear regression models for $\mathbb{E}[Y^1| X = x]$ and $\mathbb{E}[Y^0| X = x]$; in the second stage, we estimate the CATE by taking the difference of the two estimators.
\item[The X-learner \citep{Kunzel2019}:] This method modifies the T-learner by correcting the estimator using the propensity score $p(D=1|X)$.
\item[The IPW-learner \citep{saito2021open}:] This approach uses a propensity score $p(D=1|X)$ to construct a conditionally unbiased estimator of $f_0(X_i)$ as $\frac{\mathbbm{1}[D_i = 1]Y_i}{p(D_i = 1| X_i)} - \frac{\mathbbm{1}[D_i = 0]Y_i}{p(D_i = 0| X_i)}$. If $p(D_i = 1 | X_i)$ is unknown, we estimate it in some way. Then, we regress $X_i$ on the estimated $f_0(X_i)$.
\item[The DR-learner \citep{Kennedy2020}:] This learner estimates $f_0(X_i)$ by a DR estimator defined as $\frac{\mathbbm{1}[D_i = 1]\left(Y_i- \widehat{\mathbb{E}}[Y^1_i | X_i]\right)}{\hat{p}(D_i = 1| X_i)} - \frac{\mathbbm{1}[D_i = 0]\left(Y_i- \widehat{\mathbb{E}}[Y^0_i | X_i]\right)}{\hat{p}(D_i = 0| X_i)}$
to estimate $f_0(X_i)$, where $\widehat{\mathbb{E}}[Y^d | X]$ is an estimator of $\mathbb{E}[Y^d| X]$, and $\hat{p}(D_i = 1| X_i)$ is an estimator of $p(D_i = 1 | X_i)$. Then, we regress $X_i$ on the estimated $f_0(X_i)$.
\item[The R-learner \citep{Nie2020}:] This approach employs the Robinson decomposition \citep{Robinson1988} to estimate CATEs.
\end{description}

Many methods for CATE estimation can be categorized into a meta-learner. For instance, our CATE Lasso is an instance of the T-learner. 

The IPW-learner and DR-learner extend the IPW estimator \citep{Horvitz1952} and DR estimator \citep{bang2005drestimation}, originally proposed for ATE estimation, to CATE estimation.

In the context of model selection, the IPW-learner and the DR-learner have garnered attention as in \citet{schuler2018comparison} and \citet{Saito2020}. \citet{ninomiya2022information} and \citet{ninomiya2021selective} combine the IPW-learner with the Lasso. As existing studies have pointed out, the advantages of the IPW-learner stem from the unbiasedness to the risk for $Y^1 - Y^0$, while the T-learner combines two risks for $Y^1$ and $Y^0$ separately. However, the IPW-learner requires the true value for the propensity score $p(D = 1 | X)$. When it is unknown and replace it with an estimate, we cannot enjoy the unbiasedness. Furthermore, under high-dimensional models, the estimation of $p(D = 1 | X)$ itself becomes problematic because we need to assume something like sparsity for estimating $p(D = 1 | X)$ with high-dimensional $X$. The DR-learner further requires an estimate of $\mathbb{E}[Y^d|X]$ to estimate $\mathbb{E}[Y^1- Y^0|X]$, in addition to $p(D = 1 | X)$. In contrast, our method does not suffer from the problem.

\section{Experiments}
\label{sec:exp}
To verify the soundness of the CATE Lasso, we exhibit simulation studies in this section.  We compare our method with the T-Learners using the OLS and Lasso.


First, for regression models' parameters $\bm{\beta}^d$ for $d\in\mathcal{D}$, we generate each element of its first $s_0$ elements $\bm{\beta}^d_1, \dots, \bm{\beta}^d_{s_0}$ from the uniform distribution whose support is $[-10, 10]$, setting the other elements are zero, $\bm{\beta}^d_{s_0 + 1}= \cdots = \bm{\beta}^d_{p} = 0$. We also generate $\bm{\theta}$ from the uniform distribution with a support $[-1, 1]$, which is used to construct $p(D = 1| X)$. 

Let $n$, $p$, and $s_0$ be the sample size, dimension of $X$, and sparsity parameter, defined later. 
Then, we generate $\{(X_i, D_i, Y_i)\}^n_{i=1}$ as follows. We generate $V_i$ from the $p-1$-dimensional standard normal distribution and define $X_i = (1\ \  V^\top_i)^\top$. For $d\in\mathcal{D}$, let us generate $\epsilon^d_i$ from the standard distribution. Then, we obtain $Y^d_i$ as $Y_i = X^\top_i\bm{\beta}^d$. We also construct $p(D = 1 | X = X_i)$ as $p(D = 1| X = X_i) = \frac{1}{1 + \exp(-X^\top_i \bm{\theta} + \eta_i)}$, where $\eta_i$ is generated from the standard distribution. From a Bernoulli distribution with a parameter $p(D = 1 | X = X_i)$, we obtain $D_i \in \mathcal{D}$. Thus, we obtain $(X_i, D_i, Y_i)$ and obtain $\{(X_i, D_i, Y_i)\}^n_{i=1}$ by generating $(X_i, D_i, Y_i)$ $n$ times independently. 

By using $\{(X_i, D_i, Y_i)\}^n_{i=1}$, we predict $Y^1_i - Y^0_i$ given $X_i$. We conduct experiments for $n = 500$, $p = 300, 1000$, and $s_0 = 10, 25, 50$. When $p = 300$, $p > n$ does not hold, but we use this setting for confirming the effectiveness of the CATE Lasso.

We conduct each experiment $100$ times and show the root mean squared errors of the CATE Lasso, the OLS, and the Lasso in Figure~\ref{fig:exp1} and \ref{fig:exp2}. 

We can confirm that the CATE Lasso shows preferable performances in all cases. Because we do not assume sparsity for each potential outcome, the Lasso does not perform well. As $s_0$ increases, the performance of the CATE Lasso approaches to the OLS, which is an expected behavior because as $s_0$ increases, we cannot exploit the implicit sparsity. 

We also show the additional results, such as experiments with the IPW-Learner, and a semi-synthetic datasets in Appendix~\ref{appdx:exp}.

\section{Conclusion}
We examined CATE estimation using a high-dimensional linear regression model. By assuming implicit sparsity, which arises from the difference between two potential outcome linear regression models, we proposed the CATE Lasso estimators. For these estimators, we demonstrated several theoretical properties, such as consistency. Subsequently, we presented experimental results to validate our proposed estimators. Our proposed estimators represent novel, simple, and practical approaches for CATE estimation. An open issue for future work is confidence intervals and the semiparametric efficiency of these estimators.

\bibliographystyle{iclr2024_conference}
\bibliography{arXiv.bbl}

 \clearpage

\onecolumn
\appendix

\section{Theoretical Results with Fixed Design}
\label{sec:theoretical_fix}
We defined our CATE Lasso estimator in the previous section. This section provides several theoretical properties of our CATE Lasso estimator with fixed design, i.e., $\widetilde{\bm{X}}$ and $\mathbb{D}$ are non-random variables. Our analysis in this section and Sections basically follow those in \citet{vandeGeer2014}. 

Following Section~6.2 in \citet{Buhlmann2011}, we investigate the convergence of the CATE Lasso estimator $\widehat{\bm{\beta}}$ for $\bm{\beta}$ and the estimator for $f_0(x)$. We derive an upper bound for the estimation error under an appropriately chosen $\lambda$, namely \emph{oracle inequality}, which implies the consistency for the estimator.

Because the choice of $\delta$ does not affect the optimization\footnote{We can include the choice of $\delta$ in the choice of $\lambda$}, we omit it in the proof and consider the following equivalent optimization problem:
\begin{align*}
\widehat{\bm{\beta}} &:= \argmin_{\bm{\beta}\in\mathbb{R}^p} \left\{
\left\|\widetilde{\mathbb{Y}}^1 - \widetilde{\bm{X}}^1\left(\bm{\beta}  + \widehat{\bm{\beta}}^0\right) \right\|^2_2 + \lambda\left\|\bm{\beta}\right\|_1\right\}.
\end{align*}

\paragraph{Basic inequality.}
As a preliminary, we show the following basic inequality, as well as Lemma~6.1 in \citet{Buhlmann2011}. This inequality is utilized in analyses in the following parts. 

\begin{lemma}[Basic inequality. Corresponding to Lemma~6.1 in \citet{Buhlmann2011}] \label{lem:basic_ineq}
    \begin{align*}
        &\left\| \widetilde{\bm{X}}^1\left(\widehat{\bm{\beta}} - \bm{\beta}_0\right) \right\|^2_2 / n  + \lambda  \left\| \widehat{\bm{\beta}} \right\|_1\leq 2\left(\widetilde{\bm{\varepsilon}}^\dagger\right)^\top \widetilde{\bm{X}}^1\left(\widehat{\bm{\beta}} - \bm{\beta}_0\right) / n  + \lambda \|\bm{\beta}_0\|_1,
    \end{align*}
    where 
    \begin{align*}
        \widetilde{\bm{\varepsilon}}^\dagger = \bm{\varepsilon}^1 - \widetilde{\bm{X}}^1\left(\left(\widetilde{\bm{X}}^0\right)^\top\widetilde{\bm{X}}^0\right)^{\dagger}\left(\widetilde{\bm{X}}^0\right)^\top\bm{\varepsilon}^0.
    \end{align*}
\end{lemma}
\begin{proof}
Recall that we estimate $\widehat{\bm{\beta}}$ as
\begin{align*}
\widehat{\bm{\beta}} &:= \argmin_{\bm{\beta}\in\mathbb{R}^p} \left\{
\left\|\widetilde{\mathbb{Y}}^1 - \widetilde{\bm{X}}^1\left(\bm{\beta}  + \widehat{\bm{\beta}}^0\right) \right\|^2_2 + \lambda\left\|\bm{\beta}\right\|_1\right\}\\
&= \argmin_{\bm{\beta}\in\mathbb{R}^p} \left\{\left\|\widetilde{\mathbb{Y}}^1 - \widetilde{\bm{X}}^1\left(\left(\widetilde{\bm{X}}^0\right)^\top\widetilde{\bm{X}}^0\right)^{\dagger}\left(\widetilde{\bm{X}}^0\right)^\top\widetilde{\mathbb{Y}}^0 - \widetilde{\bm{X}}^1\bm{\beta} \right\|^2_2 + \lambda\left\|\bm{\beta}\right\|_1\right\}\\
&= \argmin_{\bm{\beta}\in\mathbb{R}^p} \left\{\left\| \widetilde{\bm{X}}^1\left\{\bm{\beta}^1_0 - \bm{\beta}^0_0\right\} + \bm{\varepsilon}^1 - \widetilde{\bm{X}}^1\left(\left(\widetilde{\bm{X}}^0\right)^\top\widetilde{\bm{X}}^0\right)^{\dagger}\left(\widetilde{\bm{X}}^0\right)^\top\bm{\varepsilon}^0 - \widetilde{\bm{X}}^1\bm{\beta} \right\|^2_2 + \lambda\left\|\bm{\beta}\right\|_1\right\}\\
&= \argmin_{\bm{\beta}\in\mathbb{R}^p} \left\{\left\| \widetilde{\bm{X}}^1\bm{\beta}_0 + \widetilde{\bm{\varepsilon}}^\dagger - \widetilde{\bm{X}}^1\bm{\beta} \right\|^2_2 + \lambda\left\|\bm{\beta}\right\|_1\right\}.
\end{align*}
Since $\widehat{\bm{\beta}}$ minimizes the objective function, we obtain the following inequality
\begin{align*}
&\left\| \widetilde{\bm{X}}^1\bm{\beta}_0 + \widetilde{\bm{\varepsilon}}^\dagger - \widetilde{\bm{X}}^1\widehat{\bm{\beta}} \right\|^2_2 + \lambda\left\|\widehat{\bm{\beta}}\right\|_1\leq \left\| \widetilde{\bm{X}}^1\bm{\beta}_0 + \widetilde{\bm{\varepsilon}}^\dagger - \widetilde{\bm{X}}^1\bm{\beta}_0 \right\|^2_2 + \lambda\left\|\bm{\beta}_0\right\|_1.
\end{align*}
Then, a simple calculation yields the statement.
\end{proof}

The term 
\[2\left(\widetilde{\bm{\varepsilon}}^\dagger\right)^\top \widetilde{\bm{X}}^1\left(\widehat{\bm{\beta}} - \bm{\beta}_0\right) / n,\]
in Lemma \ref{lem:basic_ineq} is referred to as the \emph{empirical process} part of this problem \citep{Buhlmann2011}. By bounding the empirical process part, 
we can bound $\| \widetilde{\bm{X}}^1(\widehat{\bm{\beta}} - \bm{\beta}_0) \|^2_2$ and $\lambda  \| \widehat{\bm{\beta}} - \bm{\beta}_0 \|_1$ from the basic inequality. 

\paragraph{Concentration inequality regarding the error term.}
For each $j\in\mathcal{P}$, to bound $(\max_{1\leq j \leq p} 2 |( \widetilde{\bm{\varepsilon}}^\dagger)^\top \widetilde{\mathbb{X}}^1_j|)\| \widehat{\bm{\beta}} - \bm{\beta}_0 \|_1$, we introduce the following event:
\begin{align*}
    \mathcal{F}_{\lambda_0} := \left\{ \max_{1\leq j \leq p} 2 \left|\left(\widetilde{\bm{\varepsilon}}^\dagger\right)^\top \widetilde{\mathbb{X}}^1_j\right| / n \leq \lambda_0\right\},
\end{align*}
for arbitrary $2\lambda_0 \leq \lambda$. For a suitable value of $\lambda_0$ and sub-Gaussian errors $\widetilde{\bm{\varepsilon}}^\dagger = (\widetilde{{\varepsilon}}^\dagger_i)_{i=1}^n$, we show that the event $\mathcal{F}_{\lambda_0}$ holds with large probability. Let us denote the diagonal elements of the Gram matrix $\widehat{\Sigma}^1 := (\widetilde{\bm{X}}^1)^\top \widetilde{\bm{X}}^1 / n$, by $\left(\widehat{\sigma}^1_j\right)^2 := \widehat{\Sigma}^1_{j, j}$, for $j\in\mathcal{P}$. Then, we show the following lemma.
\begin{lemma}[Corresponding to Lemma 6.2. in \citet{Buhlmann2011}]
\label{lem:concent}
Suppose that $\varepsilon_i$ follows a sub-Gaussian distribution with a variance $(\sigma_\varepsilon)^2$.
    Also suppose that $\left(\widehat{\sigma}^1_j\right)^2 \leq M$ for all $j \in \mathcal{P}$ and some $0 < M < \infty$. Then, for any $t > 0$ and for 
    \begin{align*}
        \lambda_0 := 2M\sigma^\dagger_\varepsilon  \sqrt{\frac{t^2 + 2 \log p}{n}},
    \end{align*}
    we have
    \begin{align*}
        \mathbb{P}\left(\mathcal{F}_{\lambda_0}\right) \geq 1 - 2\exp\left( -t^2 / 2 \right),
    \end{align*}
    where
    \begin{align*}
        (\sigma^\dagger_\varepsilon)^2 := (\sigma^1_\varepsilon)^2 +  \mathrm{tr}
        \left(\{(\widetilde{\bm{X}}^0)^\top \widetilde{\bm{X}}^0\}^\dagger(\widetilde{\bm{X}}^1)^\top \widetilde{\bm{X}}^1\right)(\sigma^0_\varepsilon)^2
    \end{align*}
\end{lemma}
\begin{proof}
    We first rewrite the term as
    \begin{align*}
        &\left(\widetilde{\bm{\varepsilon}}^\dagger\right)^\top \widetilde{\mathbb{X}}^1_j / \sqrt{n \left(\sigma^\dagger_\varepsilon\right)^2\left(\widehat{\sigma}^1_j\right)^2}\\
        & = \sum^n_{i=1} {\varepsilon}_i\mathbbm{1}[D_i = 1]X^1_{i, j} / \sqrt{n \left(\sigma^\dagger_\varepsilon\right)^2\left(\widehat{\sigma}^1_j\right)^2} +  \widetilde{\bm{X}}^1\left(\left(\widetilde{\bm{X}}^0\right)^\top\widetilde{\bm{X}}^0\right)^{\dagger}\sum^n_{i=1}\left(\widetilde{\bm{X}}^0\right)^\top\varepsilon_i\mathbbm{1}[D_i = 0]X^1_{i, j} / \sqrt{n \left(\sigma^\dagger_\varepsilon\right)^2\left(\widehat{\sigma}^1_j\right)^2}.
    \end{align*}
    Because $\varepsilon^\dagger_i$ follows a sub-Gaussian distribution with the fixed-design where recall that $D_i$ and $X^d_{i, j}$ are non-random,     
    the rewritten term also follows a sub-Gaussian distribution and achieve
    \begin{align*}
        &\mathbb{P}\left(\max_{1\leq j \leq p} \left|\left(\widetilde{\bm{\varepsilon}}^\dagger\right)^\top \widetilde{\mathbb{X}}^1_j\right| / \sqrt{n \left(\sigma^\dagger_\varepsilon\right)^2\left(\widehat{\sigma}^1_j\right)^2}\leq  \sqrt{t^2 + 2\log p}\right)\\
        &\leq \sum_{j\in\mathcal{P}} \mathbb{P}\left(\left|\left(\widetilde{\bm{\varepsilon}}^\dagger\right)^\top \widetilde{\mathbb{X}}^1_j\right| / \sqrt{n \left(\sigma^\dagger_\varepsilon\right)^2\left(\widehat{\sigma}^1_j\right)^2}\leq  \sqrt{t^2 + 2\log p}\right)\\
        &\leq 2 p \exp\left( - \frac{t^2 + 2\log p}{2} \right) = 2 \exp\left(- \frac{t^2}{2}\right),
    \end{align*}
    where the last inequality follows from the definition of a sub-Gaussian random variable. 
\end{proof}

Note that from Lemma~\ref{lem:concent}, for any $t > 0$ and for 
    \begin{align*}
        \lambda := M\sigma^\dagger_\varepsilon \sqrt{\frac{t^2 + 2 \log p}{n}}.
    \end{align*}
    with probability $1 - 2\exp(-t^2/2)$, we have
\[\left\| \widetilde{\bm{X}}^1\left(\widehat{\bm{\beta}} - \bm{\beta}_0\right) \right\|^2_2 / n \leq \left\| \widetilde{\bm{X}}^1\left(\widehat{\bm{\beta}} - \bm{\beta}_0\right) \right\|^2_2 / n  + 2 \lambda  \left\| \widehat{\bm{\beta}} \right\|_1 \leq \lambda\left\|  \widehat{\bm{\beta}} - \bm{\beta}_0 \right\|_1 / 2 + \lambda \|\bm{\beta}_0\|_1,\]
which implies
\[\left\| \widetilde{\bm{X}}^1\left(\widehat{\bm{\beta}} - \bm{\beta}_0\right) \right\|^2_2 / n \leq 3\lambda \|\bm{\beta}_0\|_1.\]

\paragraph{Compatibility condition.}
To derive a tight upper bound for the estimation error, we follow \citet{Buhlmann2011} in asserting that the compatibility condition plays a crucial role in identifiability. To define the compatibility condition, for a $p\times 1$ vector $\bm{\beta}$ and a subset $\mathcal{S}_0 \subseteq \mathcal{P}$, we introduce a notation that denotes the sparsity of $\bm{\beta}_0$. Let $\mathcal{S}_0 \subseteq \mathcal{P}$ be a set such that
\begin{align*}
    \mathcal{S}_0 := \left\{j: \beta_{0, j} \neq 0\right\},
\end{align*}
where $\beta_{0, j}$ is the $j$-th element of $\bm{\beta}_0$. 
For an index set $\mathcal{S}_0 \subset \mathcal{P}$, let $\bm{\beta}_{\mathcal{S}_0}$ and $\bm{\beta}_{\mathcal{S}^c_0}$ be vectors whose $j$-th elements are defined as follows:
\begin{align*}
    \beta_{j, \mathcal{S}_0} := \beta_{j}\mathbbm{1}[j\in \mathcal{S}_0],\qquad \beta_{j, \mathcal{S}^c_0} := \beta_{j}\mathbbm{1}[j\notin \mathcal{S}_0],
\end{align*}
respectively, where $\beta_j$ is the $j$-th element of $\bm{\beta}$. Therefore, $\bm{\beta} = \bm{\beta}_{\mathcal{S}_0} + \bm{\beta}_{\mathcal{S}^c_0}$ holds. 

Note that from the definition of the individual and common parameters, it holds that
\begin{align*}
    \beta_{j, \mathcal{S}_0} = \alpha^1_j - \alpha^0_j\quad \mathrm{and}\quad  \beta_{j, \mathcal{S}^c_0} = \gamma_j - \gamma_j = 0.
\end{align*}

We also provide the compatibility condition, which is commonly used in the analysis for Lasso.
\begin{definition}[Compatibility condition. From (6.4) in \citet{Buhlmann2011}]
    We say that the compatibility condition holds for the set $S_0$, if there is a positive constant $\phi_0 > 0$ such that for all $\bm{\beta}$ satisfying $\|\bm{\beta}_{\mathcal{S}^c_0}\|_1 \leq 3 \|\bm{\beta}_{\mathcal{S}_0}\|_1$, it holds that
    \begin{align*}
        \|\bm{\beta}_{\mathcal{S}_0}\|^2_1 \leq s_0 \bm{\beta}^\top \widehat{\Sigma} \bm{\beta} / \left(\phi_0\right)^2.
    \end{align*}
    We refer to $\phi_0$ as the compatibility constant.
\end{definition}

\paragraph{Oracle inequality and consistency.}
Lastly, we derive an oracle inequality for the CATE Lasso estimator. 
Using the compatibility condition, we make Assumption~\ref{asm:a1_vandegeer}. 

Then, we show the following oracle inequality and consistency result in Theorem~\ref{thm:oracle} with its proof in Appendix~\ref{appdx:oracle_inequality}.


\section{Proof of Theorem~\ref{thm:oracle}}
\label{appdx:oracle_inequality}
To prove Theorem~\ref{thm:oracle}, we show the following lemma, which directly yields the statements in Theorem~\ref{thm:oracle}.
\begin{lemma}
\label{lem:oracle}
    Under the same conditions in Theorem~\ref{thm:oracle}, we have
    \begin{align*}
        \left\|\widetilde{\bm{X}}^1\left(\widehat{\bm{\beta}} - \bm{\beta}_0\right) \right\|^2_2 / n + \lambda\left\|\widehat{\bm{\beta}} - \bm{\beta}_0\right\|_1 \leq 4 \lambda^2 s_0 / \left( \phi_0 \right)^2.
    \end{align*}
\end{lemma}
\begin{proof}
On $\mathcal{F}_{\lambda_0}$, by the basic inequality, for $\lambda \geq 2 \lambda_0$, we have
\begin{align}
\label{eq:target1111}
    &2\left\| \widetilde{\bm{X}}^1\left(\widehat{\bm{\beta}} - \bm{\beta}_0\right) \right\|^2_2 / n + 2\lambda  \left\| \widehat{\bm{\beta}}\right\|_1 \leq \lambda\left\| \widehat{\bm{\beta}} - \bm{\beta}_0 \right\|_1 + 2\lambda \|\bm{\beta}_0\|_1. 
\end{align}

On the last term in the LHS of \eqref{eq:target1111}, using the triangle inequality and the property $\|\widehat{\bm{\beta}}\|_1 = \|\widehat{\bm{\beta}}_{\mathcal{S}_0}\|_1 + \|\widehat{\bm{\beta}}_{\mathcal{S}_0^c}\|_1$, we have
\begin{align*}
    \left\| \widehat{\bm{\beta}}\right\|_1 &\geq \left\| \bm{\beta}_{0, \mathcal{S}_0}\right\|_1 - \left\| \widehat{\bm{\beta}}_{\mathcal{S}_0} - \bm{\beta}_{0, \mathcal{S}_0}\right\|_1 + \left\| \widehat{\bm{\beta}}_{\mathcal{S}^c_0}\right\|_1.
\end{align*}

On the first term in the RHS of \eqref{eq:target1111}, we have
\begin{align*}
    \left\| \widehat{\bm{\beta}} - \bm{\beta}_0\right\|_1 &= \left\| \widehat{\bm{\beta}}_{\mathcal{S}_0} - \bm{\beta}_{0, \mathcal{S}_0} \right\|_1 + \left\| \widehat{\bm{\beta}}_{\mathcal{S}_0^c} \right\|_1.
\end{align*}
Note that $\|\bm{\beta}_{0, \mathcal{S}^c_0}\|_1 = 0$. 
Then, we obtain 
 \begin{align*}
        &2\left\| \widetilde{\bm{X}}^1\left(\widehat{\bm{\beta}} - \bm{\beta}_0\right) \right\|^2_2 / n +2\lambda\left\| \bm{\beta}_{0, \mathcal{S}_0}\right\|_1 - 2\lambda\left\| \bm{\beta}_{0, \mathcal{S}_0} - \widehat{\bm{\beta}}_{\mathcal{S}_0} \right\|_1  + 2\lambda\left\| \widehat{\bm{\beta}}_{\mathcal{S}^c_0}\right\|_1\\
        & \leq \lambda\left\| \widehat{\bm{\beta}}_{\mathcal{S}_0} - \bm{\beta}_{0, \mathcal{S}_0} \right\|_1 +  \lambda\left\|\widehat{\bm{\beta}}_{\mathcal{S}^c_0} \right\|_1 + 2\lambda\left\| \bm{\beta}_{0, \mathcal{S}_0}\right\|_1 + 2\lambda\left\| \bm{\beta}_{0, \mathcal{S}^c_0}\right\|_1\\
        & = \lambda\left\| \widehat{\bm{\beta}}_{\mathcal{S}_0} - \bm{\beta}_{0, \mathcal{S}_0} \right\|_1 +  \lambda\left\|\widehat{\bm{\beta}}_{\mathcal{S}^c_0} \right\|_1 + 2\lambda\left\| \bm{\beta}_{0, \mathcal{S}_0}\right\|_1.
\end{align*}
Therefore, we have
  \begin{align*}
        &2\left\| \widetilde{\bm{X}}^1\left(\widehat{\bm{\beta}} - \bm{\beta}_0\right) \right\|^2_2 / n + \lambda  \left\| \widehat{\bm{\beta}}_{\mathcal{S}^c_0}\right\|_1 \leq 3\lambda\left\| \widehat{\bm{\beta}}_{\mathcal{S}_0} - \bm{\beta}_{0, \mathcal{S}_0} \right\|_1.
 \end{align*}

Therefore, we bound $\left\|\widetilde{\bm{X}}^1\left(\widehat{\bm{\beta}} - \bm{\beta}_0\right) \right\|^2_2 / n + \lambda\left\|\widehat{\bm{\beta}} - \bm{\beta}_0\right\|_1$ as
    \begin{align*}
        &2\left\|\widetilde{\bm{X}}^1\left(\widehat{\bm{\beta}} - \bm{\beta}_0\right) \right\|^2_2 / n + \lambda\left\|\widehat{\bm{\beta}} - \bm{\beta}_0\right\|_1\\
        &\leq 2\left\|\widetilde{\bm{X}}^1\left(\widehat{\bm{\beta}} - \bm{\beta}_0\right) \right\|^2_2 / n + \lambda\left\|\widehat{\bm{\beta}}_{\mathcal{S}^c_0}\right\|_1 + \lambda\left\|\widehat{\bm{\beta}}_{\mathcal{S}_0} - \bm{\beta}_{0, \mathcal{S}_0}\right\|_1\\
        &\leq 4\lambda\left\|\widehat{\bm{\beta}}_{\mathcal{S}_0} - \bm{\beta}_{0, \mathcal{S}_0}\right\|_1.
    \end{align*}
Lastly, from the compatibility condition, we have
\begin{align*}
    &4\lambda\left\|\widehat{\bm{\beta}}_{\mathcal{S}_0} - \bm{\beta}_{0, \mathcal{S}_0}\right\|_1\\
    &\leq  4\lambda \sqrt{s_0}\left\|\widetilde{\bm{X}}^1\left(\widehat{\bm{\beta}}^1 - \bm{\beta}_0^1\right) \right\|_2  / \left( \sqrt{n} \phi_0 \right)\\
    &\leq  \left\|\widetilde{\bm{X}}^1\left(\widehat{\bm{\beta}} - \bm{\beta}_0\right) \right\|^2_2  / n + 4\lambda^2 s_0 / \left( \phi_0 \right)^2,
\end{align*}
where we used $4 uv \leq u^2 + 4v^2$. Therefore, we obtain 
\begin{align*}
    2\left\|\widetilde{\bm{X}}^1\left(\widehat{\bm{\beta}} - \bm{\beta}_0\right) \right\|^2_2 / n + \lambda\left\|\widehat{\bm{\beta}} - \bm{\beta}_0\right\|_1 \leq \left\|\widetilde{\bm{X}}^1\left(\widehat{\bm{\beta}} - \bm{\beta}_0\right) \right\|^2_2  / n + 4\lambda^2 s_0 / \left( \phi_0 \right)^2.
\end{align*}
Thus, from this inequality, the statement holds. 
\end{proof}

\section{Related Work and Open Issues}
\label{appdx;related}

\subsection{Related Work}
Early work on CATEs are \citet{Heckman1997} and \citet{Heckman2005}. \citet{LeeWhang2009} and \citet{Hsu2017} discuss both the estimation and hypothesis testing of the CATE. \citet{Cai2017,Cai2021} also study confidence intervals for high-dimensional cases. \citet{Abrevaya2015} discusses the nonparametric identification of the CATE and proposes the Nadaraya-Watson-based estimator.

With the rise of machine learning algorithms, various methods for CATE estimation have been proposed. Some of them are summarized as meta-learners by \citet{Kunzel2019} (see the following section). 

Neural networks have garnered attention as a method for nonparametric estimation \citep{SchmidtHieber2020}. In causal inference, there is a stream of work that employs neural networks \citep{Johansson2016,Shalit2017,Shi2019,Hassanpour2020Learning,curth2021nonparametric,curth2021inductive}, utilizing methods and properties of neural networks, such as representation learning \citep{bengio2014representation} and multi-task learning \citep{Caruana1997}. \citet{yoon2018ganite} applies the generative adversarial nets for CATE estimation.

Furthermore, methods that utilize Gaussian processes \citep{Alaa2017,Alaa2018}, deep kernel learning \citep{Zhang2020}, boosting, tree-based methods \citep{Zeileis2008,Su2009,ImaiSt2011,Kang2012,Lipkovich2011,loh2012,Wager2018,Athey2019,Chatla2020}, nearest neighbor matching, series estimation, and Bayesian additive regression trees have been developed \citep{Hill2011}. \citet{Gunter2011}, \citet{ImaiSt2011}, and \citet{Imai2013} formulate the CATE estimation problem as a variable selection problem. Numerous methods employing machine learning algorithms have also been proposed \citep{Li2017,Kallus2017,Powers2017,Subbaswamy2018,Zhao2019,Nie2020,Kennedy2020,Hahn2020}.

Finally, other literature on high-dimensional linear regression warrants mention. Beyond the Lasso estimator, numerous approaches for high-dimensional linear regression have been proposed. Under sparsity, methods such as the Ridge \citep{bhlitem137258} and Elastic Net \citep{ZouHastie2005} have been developed. There are also high-dimensional regression approaches that do not involve regularization, known as ridgeless estimation or interpolating estimators \citep{Bartlett2020}. \citet{Bartlett2020} develops the benign-overfitting framework for the interpolating estimator, and \citet{Tsigler2023} demonstrates benign overfitting in ridge regression.

\subsection{Meta-learners}
Certain CATE estimators can be categorized into a meta-learner. Representative meta-learners are listed below:
\begin{description}[topsep=0pt, itemsep=0pt, partopsep=0pt, leftmargin=*]
\item[The S-learner \citep{Kunzel2019}:] This approach estimates $\mathbb{E}[Y | X, D]$. Using the estimator $\widehat{\mathbb{E}}[Y| X, D]$, we estimate the CATE as $\widehat{\mathbb{E}}[Y| X = x, D = 1] - \widehat{\mathbb{E}}[Y| X = x, D = 0]$.
\item[The T-learner \citep{Kunzel2019}:] This method consists of a two-step procedure: in the first stage, we separately estimate the parameters of linear regression models for $\mathbb{E}[Y^1| X = x]$ and $\mathbb{E}[Y^0| X = x]$; in the second stage, we estimate the CATE by taking the difference of the two estimators.
\item[The X-learner \citep{Kunzel2019}:] This method modifies the T-learner by correcting the estimator using the propensity score $p(D=1|X)$.
\item[The IPW-learner \citep{saito2021open}:] This approach uses a propensity score $p(D=1|X)$ to construct a conditionally unbiased estimator of $f_0(X_i)$ as $\frac{\mathbbm{1}[D_i = 1]Y_i}{p(D_i = 1| X_i)} - \frac{\mathbbm{1}[D_i = 0]Y_i}{p(D_i = 0| X_i)}$. If $p(D_i = 1 | X_i)$ is unknown, we estimate it in some way. Then, we regress $X_i$ on the estimated $f_0(X_i)$.
\item[The DR-learner \citep{Kennedy2020}:] This approach estimates $f_0(X_i)$ by using the DR estimator defined as 
\[\frac{\mathbbm{1}[D_i = 1]\left(Y_i- \widehat{\mathbb{E}}[Y^1_i | X_i]\right)}{\hat{p}(D_i = 1| X_i)} - \frac{\mathbbm{1}[D_i = 0]\left(Y_i- \widehat{\mathbb{E}}[Y^0_i | X_i]\right)}{\hat{p}(D_i = 0| X_i)}\]
to estimate $f_0(X_i)$, where $\widehat{\mathbb{E}}[Y^d | X]$ is an estimator of $\mathbb{E}[Y^d| X]$, and $\hat{p}(D_i = 1| X_i)$ is an estimator of $p(D_i = 1 | X_i)$. Then, we regress $X_i$ on the estimated $f_0(X_i)$.
\item[The R-learner \citep{Nie2020}:] This approach employs the Robinson decomposition \citep{Robinson1988}.
\end{description}

Many methods for CATE estimation can be categorized into a meta-learner. For instance, our CATE Lasso is an instance of the T-learner. 

The IPW-learner and DR-learner extend the IPW estimator \citep{Horvitz1952} and DR estimator \citep{bang2005drestimation}, originally proposed for ATE estimation, to CATE estimation.

In the context of model selection, the IPW-learner and the DR-learner have garnered attention as in \citet{schuler2018comparison} and \citet{Saito2020}. \citet{ninomiya2022information} and \citet{ninomiya2021selective} combine the IPW-learner with the Lasso. As existing studies have pointed out, the advantages of the IPW-learner stem from the unbiasedness to the risk for $Y^1 - Y^0$, while the T-learner combines two risks for $Y^1$ and $Y^0$ separately. However, the IPW-learner requires the true value for the propensity score $p(D = 1 | X)$. When it is unknown and replace it with an estimate, we cannot enjoy the unbiasedness. Furthermore, under high-dimensional models, the estimation of $p(D = 1 | X)$ itself becomes problematic because we need to assume something like sparsity for estimating $p(D = 1 | X)$ with high-dimensional $X$. The DR-learner further requires an estimate of $\mathbb{E}[Y^d|X]$ to estimate $\mathbb{E}[Y^1- Y^0|X]$, in addition to $p(D = 1 | X)$. In contrast, our method does not suffer from the problem.

\subsection{Confidence Intervals and Efficiency}
The future direction of this study is to debias the CATE Lasso estimator to obtain confidence intervals, as well as the \emph{debiased} Lasso in \citep{vandeGeer2014}. The debiased åLasso estimator is one of the \emph{post-regularization inference} methods for Lasso-based estimators. By debiasing the parameter, we obtain confidence intervals for the Lasso-based estimators.

\citet{zhangzhang2014} introduces the debiased Lasso, and several existing studies, such as \citet{Javanmard14a} and \citet{vandeGeer2014}, extend the method. Other related work is \citet{Belloni2014}, \citet{Belloni2014Post} and \citet{Belloni2016}. Specifically, \citet{Belloni2014} considers treatment effect estimation as well as ours.

There is a problem when we develop a debiased-Lass estimator in our formulation. Following the approach of \citep{vandeGeer2014}, we can define the debiased CATE Lasso estimator $\widehat{\bm{b}}$ as follows:
\begin{align*}
\widehat{\bm{b}} &:= \widehat{\bm{\beta}} + \widehat{\Theta}^1 \left(\widetilde{\bm{X}}^1\right)^\top \left(\widetilde{\mathbb{Y}}^1 - \widetilde{\bm{X}}^1\widehat{\bm{\beta}}^1\right) / n\\
&\ \ \ \ \ \ \ \ \ \ \ \ \ \ \ \ \ \ \ \ - \widehat{\Theta}^0 \left(\widetilde{\bm{X}}^0\right)^\top \left(\widetilde{\mathbb{Y}}^0 - \widetilde{\bm{X}}^0\widehat{\bm{\beta}}^0\right) / n,
\end{align*}
where $\widehat{\Theta}^d$ is a reasonable approximation for inverses of $\widehat{\Sigma}^d := \left(\widetilde{\bm{X}}^d\right)^\top\widetilde{\bm{X}}^d / n$ for $d\in\{1, 0\}$. 
Then, $
\sqrt{n}\left(\widehat{\bm{b}} - \bm{\beta}_0\right) = \widehat{\Theta}^1 \left(\widetilde{\bm{X}}^1\right)^\top \widetilde{\bm{\varepsilon}}^1 / \sqrt{n} - \widehat{\Theta}^0 \left(\widetilde{\bm{X}}^0\right)^\top \widetilde{\bm{\varepsilon}}^0 / \sqrt{n} + \Delta$
holds, 
where $\Delta := \Delta^1 - \Delta^0$ and for each $d\in\mathcal{D}$, 
\begin{align*}
&\Delta^d := \sqrt{n}\left(\widehat{\Theta}^d\widehat{\Sigma}^d - I\right)\left(\widehat{\bm{\beta}}^d - \bm{\beta}^d_0\right).
\end{align*}
To obtain the debiased estimator $\widehat{\bm{b}}$ with $\sqrt{n}$-consistency, we need to show $\Delta^d = o_p(1)$. In the standard debiased Lasso, we obtain the result by making the product of $\left(\widehat{\Theta}^d\widehat{\Sigma}^d - I\right) = o_p(1)$ and $\left(\widehat{\bm{\beta}}^d - \bm{\beta}^d_0\right) = o_p(1)$. However, in our approach, $\left(\widehat{\bm{\beta}}^d - \bm{\beta}^d_0\right) = o_p(1)$ may not hold; that is, although $\widehat{\bm{\beta}}$ will converge to $\bm{\beta}_0$, $\widehat{\bm{\beta}}^d$ may not converge to $\bm{\beta}^d_0$ in our approach. Therefore, in our approach, it is unknown how to obtain the debiased estimator in our setting. 

Furthermore, efficiency of the estimator is also an important open issue. For example, \citet{Jankova2018} discusses the semiparametric efficiency of the debiased Lasso estimator. However, to uncover the semiparametric efficiency, we conjecture that some techniques used for semiparametric efficiency in treatment effect estimation are required, such as those described in \citet{hahn1998role}.

Among the meta-learners, the DR estimator is frequently discussed within the context of semiparametric efficient CATE estimation \citep{Fan2022}. The DR estimation has a close relationship to the semiparametric efficient influence function \citep{bickel98}. Based on this property, \citet{ChernozhukovVictor2018Dmlf} proposes double machine learning for ATE estimation, and \citet{Fan2022} applies it to semiparametric efficient CATE estimation. How the variance of our estimator relates to that of estimators with the DR estimation remains an open issue.

\subsection{Future Work: Confidence Intervals and Efficiency}
The future direction of this study is to debias the CATE Lasso estimator to obtain confidence intervals, as well as the \emph{debiased} Lasso in \citep{vandeGeer2014}. The debiased Lasso estimator is one of the \emph{post-regularization inference} methods for Lasso-based estimators. By debiasing the parameter, we obtain confidence intervals for the Lasso-based estimators.

\citet{zhangzhang2014} introduces the debiased Lasso, and several existing studies, such as \citet{Javanmard14a} and \citet{vandeGeer2014}, extend the method. Other related work is \citet{Belloni2014}, \citet{Belloni2014Post} and \citet{Belloni2016}. Specifically, \citet{Belloni2014} considers treatment effect estimation as well as ours.

There is a problem when we develop a debiased-Lass estimator in our formulation. Following the approach of \citep{vandeGeer2014}, we can define the debiased CATE Lasso estimator $\widehat{\bm{b}}$ as $
\widehat{\bm{b}} := \widehat{\bm{\beta}} + \widehat{\Theta}^1 \left(\widetilde{\bm{X}}^1\right)^\top \left(\widetilde{\mathbb{Y}}^1 - \widetilde{\bm{X}}^1\widehat{\bm{\beta}}^1\right) / n- \widehat{\Theta}^0 \left(\widetilde{\bm{X}}^0\right)^\top \left(\widetilde{\mathbb{Y}}^0 - \widetilde{\bm{X}}^0\widehat{\bm{\beta}}^0\right) / n$, 
where $\widehat{\Theta}^d$ is a reasonable approximation for inverses of $\widehat{\Sigma}^d := \left(\widetilde{\bm{X}}^d\right)^\top\widetilde{\bm{X}}^d / n$ for $d\in\{1, 0\}$. 
Then, $
\sqrt{n}\left(\widehat{\bm{b}} - \bm{\beta}_0\right) = \widehat{\Theta}^1 \left(\widetilde{\bm{X}}^1\right)^\top \widetilde{\bm{\varepsilon}}^1 / \sqrt{n} - \widehat{\Theta}^0 \left(\widetilde{\bm{X}}^0\right)^\top \widetilde{\bm{\varepsilon}}^0 / \sqrt{n} + \Delta$
holds, 
where $\Delta := \Delta^1 - \Delta^0$ and for each $d\in\mathcal{D}$, 
$\Delta^d := \sqrt{n}\left(\widehat{\Theta}^d\widehat{\Sigma}^d - I\right)\left(\widehat{\bm{\beta}}^d - \bm{\beta}^d_0\right).$
To obtain the debiased estimator $\widehat{\bm{b}}$ with $\sqrt{n}$-consistency, we need to show $\Delta^d = o_p(1)$. In the standard debiased Lasso, we obtain the result by making the product of $\left(\widehat{\Theta}^d\widehat{\Sigma}^d - I\right) = o_p(1)$ and $\left(\widehat{\bm{\beta}}^d - \bm{\beta}^d_0\right) = o_p(1)$. However, in our approach, $\left(\widehat{\bm{\beta}}^d - \bm{\beta}^d_0\right) = o_p(1)$ may not hold; that is, although $\widehat{\bm{\beta}}$ will converge to $\bm{\beta}_0$, $\widehat{\bm{\beta}}^d$ may not converge to $\bm{\beta}^d_0$ in our approach. Therefore, in our approach, it is unknown how to obtain the debiased estimator in our setting. 

Furthermore, efficiency of the estimator is also an important open issue. For example, \citet{Jankova2018} discusses the semiparametric efficiency of the debiased Lasso estimator. However, to uncover the semiparametric efficiency, we conjecture that some techniques used for semiparametric efficiency in treatment effect estimation are required, such as those described in \citet{hahn1998role}.

\begin{figure*}[h]
  \centering
    \includegraphics[width=140mm]{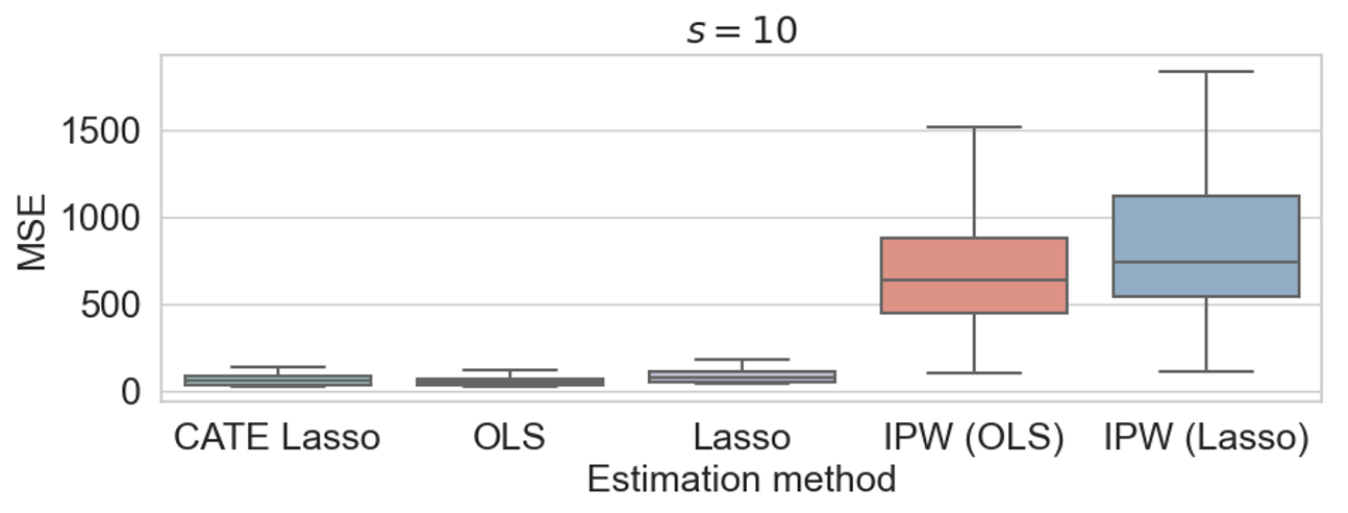}
    \includegraphics[width=140mm]{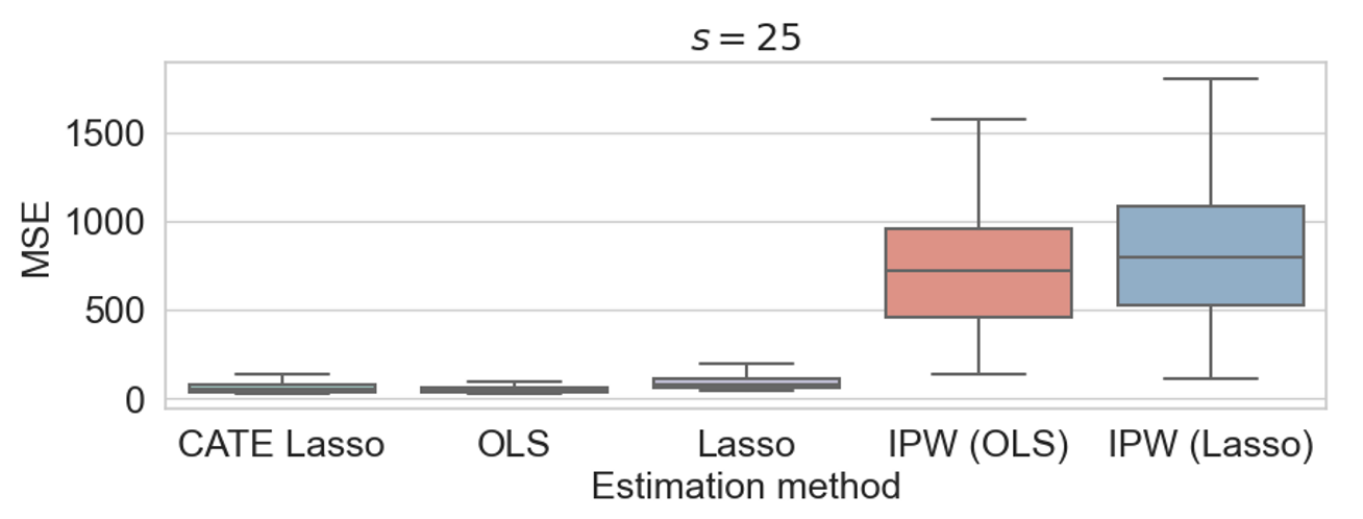}
    \includegraphics[width=140mm]{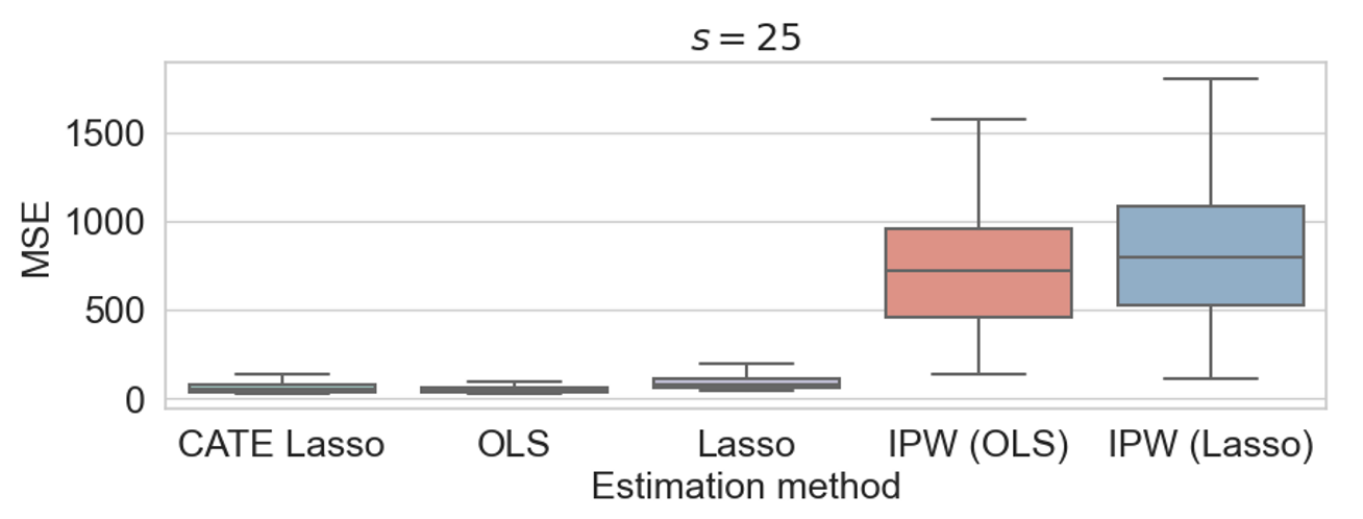}
\caption{RMSE of simulation studies.}
\label{fig:exp3}
\vspace{-5mm}
\end{figure*}

\begin{figure*}[h]
  \centering
  \vspace{-3mm}
    \includegraphics[width=140mm]{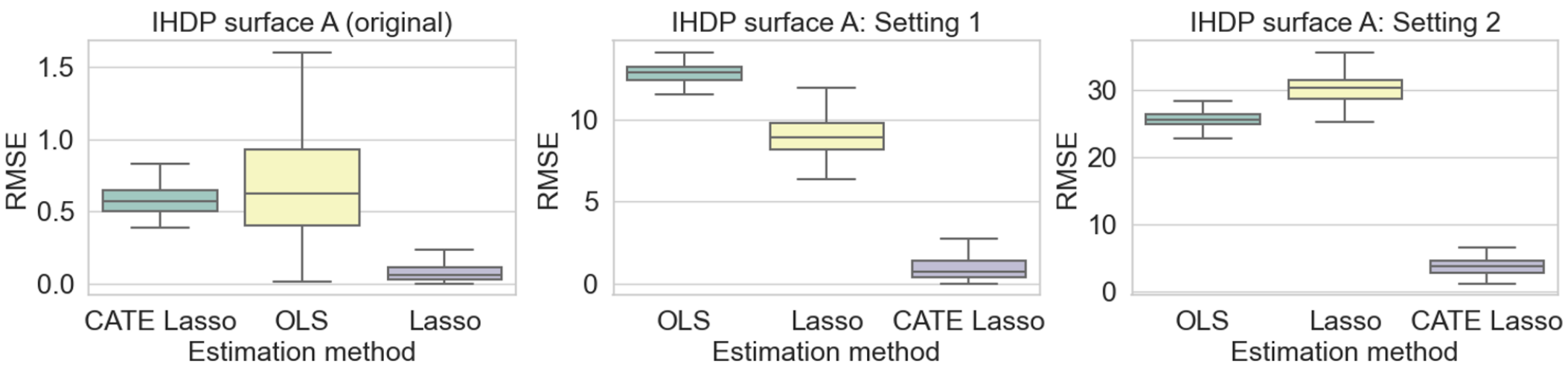}
    \caption{RMSE of experiments uing the IHDP dataset.}
\label{fig:exp5}
\end{figure*}

\section{Additional Experimental Results}
\label{appdx:exp}
This section provides additional experimental results to confirm the soundness of the CATE Lasso. In all methods based on the Lasso, we set the regularizer as $\lambda = 1.0$.

\subsection{Experiments with the IPW-Learner}
First, using the same settings in Section~\ref{sec:exp}, we investigate the performances of the IPW-Learner. Given the known propensity score $p(D_i = 1| X_i)$, we compute $\widetilde{Y}_i = \frac{\mathbbm{1}[D_i = 1]Y_i}{p(D_i = 1| X_i)} - \frac{\mathbbm{1}[D_i = 0]Y_i}{1 - p(D_i = 1| X_i)}$ and regress $\widetilde{Y}_i$ on $X_i$. In regression, we perform both the OLS and the Lasso. We show the results in Figure~\ref{fig:exp3}.

As shown in the results, the IPW-Learner does not perform well. This result is because the inverse probability makes the regression unstable. 

\subsection{Experiments with a Semi-synthetic Datasets}
In evaluating algorithms for estimating the treatment effect, it is difficult to find real-world datasets that can be used for the evaluation. Following existing work, we use semi-synthetic datasets made from the Infant Health and Development Program (IHDP), which consists of simulated outcomes and covariate data from a real study. 

We follow a setting of simulation proposed by \citet{Hill2011}, called the surface A. In the setting of \citet{Hill2011}, $747$ samples with $6$ continuous covariates and $19$ binary covariates are used. \citet{Hill2011} generated the outcomes using the covariates artificially. \citet{Hill2011} considered two scenarios: response surface A and response surface B. This study only focuses on the surface A, where $Y^1_i$ and $Y^0_i$ are generated as follows:
\begin{align*}
&Y^1_i \sim \mathcal{N}(X^\top_i\bm{\beta}_A+4, 1),\\
&Y^0_i \sim \mathcal{N}(X^\top_i\bm{\beta}_A, 1),
\end{align*}
where  elements of $\bm{\beta}_A\in\mathbb{R}^{25}$ were randomly sampled from $(0, 1, 2, 3, 4)$ with probabilities $(0.5, 0.2, 0.15, 0.1, 0.05)$.  

In addition to the original surface A, because the original covariate $X_i$ is low dimensional, we conduct experiments with additional high-dimensional covariates. We generate $Y^1_i$ and $Y^0_i$ as follows:
\begin{align*}
&Y^1_i \sim \mathcal{N}(X^\top_i\bm{\beta}_A + \widetilde{X}^\top_i\widetilde{\bm{\beta}}_A +4, 1),\\
&Y^0_i \sim \mathcal{N}(X^\top_i\bm{\beta}_A + \widetilde{X}^\top_i\widetilde{\bm{\beta}}_A, 1),
\end{align*}
where $X^\top_i\bm{\beta}_A$ is generated in the same way as the original surface A, and $Z_i$ is generated from a uniform distribution with a support $[-1, 1]$. We further consider two settings about $\widetilde{\bm{\beta}}_A$: we generate $\widetilde{\bm{\beta}}_A$ from a uniform distribution with a support $[0, 5]$ and $[0, 10]$ and refer to the settings as the \emph{setting} $1$ and \emph{setting} $2$.
 
We show the MSEs between of CATE estimation in the box-plot of Figure~\ref{fig:exp5}. We compute the MSEs by using the same $X_i$ used for estimating parameters but with the true CATE values for computing MSEs.

\end{document}